\documentclass[11pt]{article}

\usepackage{amsmath, amssymb, amsthm}
\usepackage{thmtools}
\usepackage{mathtools}
\usepackage{dsfont}
\usepackage{geometry}
\usepackage{xcolor}
\usepackage{graphicx} %
\usepackage{csquotes}
\usepackage{url}
\usepackage[colorlinks=true,urlcolor=blue!80!black,linkcolor=blue!80!black,citecolor=purple!80!black]{hyperref}
\usepackage[capitalise,nameinlink]{cleveref}
\usepackage{hyperref}
\usepackage[T1]{fontenc}
\usepackage{pifont}
\newcommand{\cmark}{\ding{51}}
\newcommand{\xmark}{\ding{55}}

\usepackage{titlesec}
\titlespacing*{\paragraph}{0pt}{0.8ex plus 0.1ex}{0.8em} %

\usepackage[
 backend=biber,
 style=alphabetic,
 citetracker,
 hyperref=auto,
 maxcitenames=5, %
 sortcites,      %
 sorting=nyt,
 maxbibnames=12, %
 date=year,      %
 isbn=false,     %
 url=false,      %
 doi=false,      %
 eprint=false,   %
]{biblatex}

\AtEveryBibitem{%
  \clearname{editor}%
  \clearfield{editortype}%
}

\newbibmacro{string+doiurlisbn}[1]{%
  \iffieldundef{doi}{%
    \iffieldundef{url}{%
      #1
    }{%
      \href{\thefield{url}}{#1}%
    }%
  }{%
    \href{http://dx.doi.org/\thefield{doi}}{#1}%
  }%
}

\DeclareFieldFormat%
[article,inbook,incollection,inproceedings,patent,thesis,unpublished,misc,book]
  {title}{\usebibmacro{string+doiurlisbn}{#1}}

\usepackage{booktabs}
\usepackage{tikz}
\usetikzlibrary{arrows.meta,calc,positioning}
\usepackage{pgfplots}
\usepgfplotslibrary{groupplots}
\usepackage{subcaption}
\usepackage{paralist}
\usetikzlibrary{arrows.meta, positioning, shadows}

\newtheorem{theorem}{Theorem}
\newtheorem{definition}[theorem]{Definition}
\newtheorem{observation}[theorem]{Observation}
\newtheorem{lemma}[theorem]{Lemma}

\newtheorem*{theorem*}{Theorem}

\crefname{observation}{Observation}{Observations}
\Crefname{observation}{Observation}{Observations}

\usepackage{xspace}
\newcommand{\PropX}{\ensuremath{\mathrm{PropX}\xspace}}
\newcommand{\PropOne}{\ensuremath{\mathrm{Prop1}\xspace}}
\newcommand{\Prop}{\ensuremath{\mathrm{Prop}}\xspace}
\newcommand{\PS}{\ensuremath{\mathrm{PS}}}
\newcommand{\N}{\ensuremath{\mathcal{J}}}
\newcommand{\M}{\ensuremath{\mathcal{M}}}

\newcommand{\bZ}{\ensuremath{\mathbb{Z}}}

\usepackage[colorinlistoftodos,prependcaption,textsize=scriptsize]{todonotes}

\title{Online Fair Division under Eligibility Constraints}

\author{
Joes Biburger\thanks{Faculty of Mathematics and Computer Science, University of Bremen, Germany. \texttt{\{biburger, nmegow, golnoosh\}@uni-bremen.de}}
\and Nicole Megow\footnotemark[1]
\and Golnoosh Shahkarami\footnotemark[1] \textsuperscript{,}\thanks{Max Planck Institute for Informatics, Germany.}
}

\date{\today}

\begin{document}

\maketitle

\begin{abstract}
We study fair division of indivisible items under eligibility constraints: each item may be assigned only to a subset of eligible agents. This model coincides with classical restricted assignment in scheduling, with $m$ agents as machines and items as jobs, and captures settings where work or resources must be split among parties with different capabilities. Eligibility makes the usual fairness benchmarks inadequate, as agents should not be measured against items they could never receive. We therefore work with eligibility-aware versions of proportionality and envy-freeness, in which an agent's proportional share is what it would receive if every item were split equally among its eligible agents, and envy comparisons discount items the agents are not eligible for. We then determine which guarantees survive.

Our main finding is a sharp chore--goods separation for proportionality. We focus on identical valuations: under arbitrary valuations, no online algorithm can guarantee even an $o(m)$-approximation of proportionality for chores. For chores, $\PropOne$ is achievable offline, whereas exact proportionality and proportionality up to \emph{any} item are not. Our main technical result is for online chores: the optimal approximation factor for $\PropOne$ achievable by a deterministic online algorithm is $\Theta(\sqrt m)$, and the lower bound already holds for unit-size jobs. For goods, $\PropOne$ is again achievable offline, but the tight online bound drops to $\Theta(\log m)$.

For envy-based fairness, the picture is different. Goods and chores become equivalent under eligibility-aware envy notions, while the choice of notion is decisive. Strong variants largely fail: strong EF1 may fail to exist even offline and cannot be approximated online. In contrast, weak EF1 can be maintained online, while weak EFX separates offline from online. Together, these results give a systematic study of online proportional and envy-based fairness in restricted assignment and show that eligibility determines which fairness guarantees remain achievable.
\end{abstract}

\thispagestyle{empty} %
\clearpage
\setcounter{page}{1}

\section{Introduction}
Fair division studies how to split a collection of items among $m$ agents with competing claims so that none is treated unfairly. The question has a long history in economics and is now a central topic in computer science and multi-agent systems. Most work concerns \emph{goods}, which agents want, but an equally natural reading treats the items as \emph{chores} that carry cost or workload; we study both. Because items are indivisible, the classical ideals of envy-freeness and proportionality can fail outright, and the field has developed meaningful relaxations such as envy-freeness up to one item (EF1) and up to any item (EFX), the maximin share guarantee (MMS), and proportionality up to one item ($\PropOne$)~\cite{LMMS04,Budish11,CKM+19,Pro20,AMNS17}. These range from mild to strong: EF1 and $\PropOne$ ask only that unfairness be repairable by setting aside a single item, while EFX and MMS ask for more.

The classical theory is \emph{offline}: The whole instance is known in advance, and the divider can plan the allocation with everything in view. Plenty of real situations refuse to cooperate. Items show up one at a time, their values are learned only on arrival, and each has to be handed out on the spot, with no chance to take it back. This \emph{online} version of the problem appears throughout practice from scheduling jobs on a shared cluster to handing out labeling tasks in a crowdsourcing pipeline to allocating donations as they arrive.
The known impossibility results are stark. With three or more agents, no online algorithm can guarantee any finite approximation to MMS, even if it is told the total value up front~\cite{ZBW23}, and the same wall appears for EF1~\cite{neoh2026online}. Online fairness is harder than offline fairness; even weak targets are in question. Exact proportionality requires each of the $m$ agents to receive at least a $1/m$ share and can fail even offline.
Since $\PropOne$ is implied by both EF1 and MMS, it is the mildest of these standard benchmarks and the natural target left standing online; even so, under general additive valuations its online approximability degrades with the number of items~\cite{DuS2026online}, whereas under identical valuations it is achievable exactly.

There is a second feature of these applications that the standard models leave out: the agents are usually not interchangeable. An item can often go only to some of them.
A specialized task needs a qualified worker, a compute job needs a machine that meets its requirements, a perishable donation can only reach recipients close enough to receive it in time. The agents who can take a given item form its \emph{eligibility set}, and any allocation that ignores these sets is not implementable. The usual fair-division setup, by contrast, assumes any item can be given to anyone.

Eligibility brings fair division face to face with classical restricted-assignment problems in scheduling: makespan minimization for chores and max-min allocation, or Santa Claus, for goods. These objectives are global, optimizing the worst load or value. Our fairness requirements are local and eligibility-aware. When machines represent independent parties, as in a shared research cluster, a cloud platform, or a volunteer-compute network, minimizing the worst load is not enough; each party wants a fair slice of the work it is actually eligible to do. Similarly, on the goods side, maximizing the minimum received value need not ensure that each agent receives a fair share of the items available to it. Thus fairness under restricted assignment is not a by-product of classical scheduling objectives, but a genuinely separate requirement.

If we read machines as agents and jobs as items, restricted assignment is exactly fair division under eligibility constraints, and this is the lens we take. Eligibility also changes what fairness means. A standard valuation-based model can encode which items an agent can receive, but it does not capture how responsibility or entitlement for an item should be shared among the agents eligible for it. An item eligible for only a few agents should count more strongly toward their proportional benchmarks than an item eligible for everyone. We therefore work with an \emph{eligibility-aware} proportional share, in which a machine's fair share is measured only against the jobs it is eligible for and the machines competing for those same jobs. Concretely, imagine every job split equally among its eligible machines; the proportional share $\mathrm{PS}_i$ of machine $i$ is the total it receives in this fractional relaxation,
\[
  \mathrm{PS}_i \;=\; \sum_{j \,:\, i \in \mathcal{M}(j)} \frac{p_j}{k_j},
\]
where $p_j$ is the size of job $j$, the same for every machine eligible for it (we work with identical valuations throughout), $\mathcal{M}(j)$ its set of eligible machines, and $k_j = |\mathcal{M}(j)|$. The shares sum to the total size of all jobs, and when every job is eligible for every machine they reduce to the usual $1/m$ share. The same share was recently and independently proposed by Sun, Wang, and Li~\cite{SWL26orientations}, who settle its offline existence questions; our focus is the online problem.

A natural alternative benchmark would be the maximin share, which under identical valuations and without eligibility constraints simply equals the optimal makespan (we phrase this for chores; goods are symmetric). Neither of the two natural ways to adapt it to eligibility constraints is adequate. If a machine spreads \emph{all} jobs as evenly as eligibility allows and takes the maximum load as its benchmark, every machine obtains the same value, the optimal makespan, and we are back at a global objective that is well understood in the restricted-assignment setting, both offline~\cite{lenstra1990approximation} and online~\cite{DBLP:journals/jacm/AspnesAFPW97}. If it spreads only the jobs it is eligible for, the benchmark may be unattainable within a factor $o(m)$ by every allocation, even when a perfectly balanced allocation exists.\footnote{Take one unit-size job for each pair of machines, eligible for exactly that pair. Each machine's benchmark is $1$, but every allocation loads some machine with $\Omega(m)$ jobs, including the one that gives each machine half of its jobs.} The proportional share has neither defect, which is why we take it as our primary benchmark: each machine is measured against its own eligibility set, and the shares of all machines partition the total load.

To complete the picture, we also ask how envy-based fairness should be read under eligibility constraints, since here, too, comparisons between machines are only meaningful through the jobs eligibility allows. We consider two versions. The strong notion compares two machines only through jobs that both could have received. The weak notion counts the full load of the envied machine but, on the envying side, ignores jobs the envied machine could not have received. Strong and weak EF1 and EFX then ask, as usual, that envy disappear after removing one, respectively any, job on the envying side.

With these eligibility-aware notions in place, we ask what can and cannot be achieved, with online $\PropOne$ as our central question. This gives a first systematic view of online fair allocation under restricted assignment. Our results show that eligibility constraints fundamentally reshape the fair-division landscape: they separate goods from chores for proportionality, but bring them back together for envy-based fairness.

\subsection{Our Results}

In the unrestricted setting, where every item is eligible for every agent, proportionality admits strong guarantees: proportionality up to any item ($\PropX$) is achievable offline, and $\PropOne$ remains achievable online.\footnote{For identical chores, assigning each chore to a currently least-loaded agent in non-increasing order of processing times gives EFX, and hence \PropX, offline. Without sorting, least-loaded assignment maintains EF1, and hence \PropOne, online.}

Eligibility changes this picture completely, and not uniformly. Throughout, we work with identical valuations; under general valuations even $o(m)$-$\PropOne$ is unattainable online for chores; see \Cref{app:impossibility-general-valuations}. Here and below, for $\alpha \ge 1$, an allocation is $\alpha$-$\PropOne$ if it satisfies the $\PropOne$ condition with the share relaxed by a factor $\alpha$, to $\alpha\,\mathrm{PS}_i$ for chores and to $\mathrm{PS}_i/\alpha$ for goods. For proportionality, eligibility creates a $\Theta(\sqrt m)$ online price for chores but only a $\Theta(\log m)$ price for goods; for envy-based fairness, it collapses the distinction between the two settings. Note that our fairness guarantees are not consequences of classical restricted-assignment objectives. Makespan and Santa Claus optimize global worst-case quantities, whereas our benchmarks are local and eligibility-aware; in fact, the corresponding scheduling and fairness guarantees cannot, in general, be achieved simultaneously, not even with constant loss (\Cref{app:impossibility-simultan-scheduling}).

\paragraph{Proportionality for chores.}
In the chore setting, each job adds to a machine's load, and fairness asks that no machine exceed its eligibility-aware share. We first record the offline baseline: exact proportionality and $\PropX$ are not always achievable, while $\PropOne$ always is; see \Cref{sec:offline}, and~\cite{SWL26orientations} for the same facts under general additive valuations. Thus, for chores under eligibility constraints, $\PropOne$ is the strongest proportionality notion that is always achievable.

The main focus of the paper is the online setting, where jobs arrive one by one and must be assigned irrevocably. We measure performance by the smallest $\alpha$ for which an algorithm guarantees $\alpha$-\PropOne{}. Without eligibility constraints, $\PropOne$ is achievable online exactly, so any $\alpha > 1$ is a price paid for eligibility alone. Our main result determines this online price of eligibility up to constant factors.

\begin{theorem*}[informal; see \Cref{sec:chore_online}]
For restricted-assignment chores with arbitrary processing times, the optimal deterministic online approximation factor for $\PropOne$ is $\Theta(\sqrt m)$; the lower bound already holds for unit-size jobs.
\end{theorem*}

\paragraph{Proportionality for goods.}
For goods, the offline baseline mirrors the chore setting: $\PropOne$ is always achievable, while exact proportionality and $\PropX$ are not guaranteed; see \Cref{sec:offline}. Online, however, the price of eligibility is much smaller than for chores.

\begin{theorem*}[informal; see \Cref{sec:goods_setting}]
For restricted-assignment goods with arbitrary values, the optimal deterministic online approximation factor for $\PropOne$ is $\Theta(\log m)$; the lower bound already holds for unit values.
\end{theorem*}

Thus the same eligibility constraints that force a $\Theta(\sqrt m)$ loss for online chores cost only $\Theta(\log m)$ for online goods.

\paragraph{Envy-based fairness.}
For envy-based fairness, the picture is different. While proportionality separates goods from chores, envy brings them back together: under identical valuations, the two settings are equivalent for eligibility-aware envy notions. We therefore state the results for chores, and they transfer directly to goods.

\begin{theorem*}[informal; see \Cref{sec:envy}]
Strong envy-based guarantees are largely impossible under restricted assignment: strong EF1 may fail to exist even offline, and no deterministic online algorithm can guarantee any approximation of it. In contrast, weak EF1 can be maintained online for general processing times, while weak EFX exists offline but cannot be maintained online.
\end{theorem*}

These results give, to the best of our knowledge, the first systematic study of online proportional and envy-based fairness under eligibility constraints. They identify the exact online price of eligibility for proportionality, reveal a sharp $\sqrt m$ versus $\log m$ chore--goods divide, and show that envy-based fairness follows a different logic: goods and chores coincide, while strong and weak notions diverge. \Cref{tab:results} summarizes the results. Where a notion is marked as not always achievable, we do not determine its best possible approximation; in particular, the optimal factors for $\alpha$-Prop and $\alpha$-$\PropX$, offline and online, remain open (\Cref{sec:conclusion}).

\begin{table}[t]
\centering
\begin{tabular}{lccc}
\toprule
 & & \multicolumn{2}{c}{Online} \\
\cmidrule(lr){3-4}
Notion & Offline & Chores & Goods \\
\midrule
Prop, $\PropX$ & \xmark & \xmark & \xmark \\
$\PropOne$ & \cmark & $\Theta(\sqrt m)$-$\PropOne$ & $\Theta(\log m)$-$\PropOne$ \\
\midrule
strong EF1, strong EFX & \xmark & \multicolumn{2}{c}{no approximation} \\
weak EFX & \cmark & \multicolumn{2}{c}{\xmark} \\
weak EF1 & \cmark & \multicolumn{2}{c}{\cmark} \\
\bottomrule
\end{tabular}
\caption{Summary of results. Online bounds are for deterministic algorithms.}
\label{tab:results}
\end{table}%

\subsection{Overview of techniques}

We briefly highlight the main ideas behind our results. The offline proportionality results are the one place where classical restricted-assignment machinery applies directly: after normalizing loads by the eligibility-aware proportional shares, the assignment-LP rounding of Lenstra, Shmoys, and Tardos~\cite{lenstra1990approximation} yields $\PropOne$ allocations for chores, and the modification by Bez{\'a}kov{\'a} and Dani~\cite{bezakova2005allocating} yields $\PropOne$ for goods. Online, however, the analysis has to track local eligibility-aware fairness slack rather than a global load or value objective, which is why standard online load-balancing arguments do not directly apply.

For online chores (\Cref{sec:chore_online}), the difficulty already appears for unit-size jobs, where it is caused purely by eligibility.
This is also where standard online load-balancing intuition first breaks down: assigning a job to a least-loaded eligible machine gives an $O(\log m)$-competitive algorithm for makespan, but the natural fairness analogue, assigning to a machine with maximum current $\PropOne$-slack, can fail to achieve an $o(m)$-approximation; see \Cref{app:greedy-fails}.
We instead use a greedy rule based on approximate proportional slack: each arriving job is assigned to an eligible machine with maximum $\alpha$-slack, where $\alpha=\Theta(\sqrt m)$.
The analysis uses a deliberately non-smooth potential, with a roughly exponential part for small positive slack and a discrete jump at zero.
This jump makes negative slack globally expensive and lets us bound how many machines can violate the approximate proportionality threshold at once.
When a job is assigned to one eligible machine, its load hurts that machine, but its proportional-share contribution helps all other eligible machines; the potential shows that, under the greedy choice, these gains compensate for the loss.

Arbitrary processing times break this simple greedy logic, since one large job can destroy the slack of the receiving machine.
The algorithm therefore works with $\alpha$-\PropOne{}-slack and distinguishes two assignment modes. Either the job \emph{fits} on a machine with enough slack to absorb it, or it is assigned to a machine whose previous largest job can be \emph{shadowed} by the new job's proportional-share contribution.
The core of the analysis is a job-dependent potential with
a local penalty that guarantees that whenever no \emph{fit}-machine exists, many eligible machines satisfy the \emph{shadow} condition. This lifts the $\Theta(\sqrt m)$ guarantee from unit-size jobs to arbitrary processing times.

For online goods (\Cref{sec:goods_setting}), the same eligibility constraints lead to a different scale. For unit values, we use the natural sign-reversed analogue of approximate slack: each arriving good is assigned to an eligible agent with maximum $\alpha$-deficit, where $\alpha=\Theta(\log m)$. The potential is correspondingly simpler: it penalizes large deficits exponentially.
The key difference from chores is that an assigned chore creates one large load increase, whereas an assigned good creates only small benchmark increases for the other eligible agents. This makes the deficits easier to control than the slacks; the exponential potential shows that the greedy choice keeps all deficits within one unit at the $\Theta(\log m)$ scale.
For arbitrary values, we group goods into value classes and run the unit-value algorithm separately in each class. The bucketing loses only another constant factor, giving the optimal $\Theta(\log m)$ guarantee.

The envy-based results (\Cref{sec:envy}) use different ideas. The main structural observation is that, unlike proportionality, eligibility-aware envy preserves the equivalence between goods and chores:
under identical valuations, the same allocation satisfies the corresponding envy relations in the other setting with reversed direction of comparison.
This equivalence is immediate without eligibility constraints, but it is no longer automatic once envy comparisons are restricted by eligibility. It nevertheless holds for our strong and weak envy notions, allowing us to prove all envy results in one setting and transfer them to the other.
The weak notion then behaves much like the unrestricted case and requires only minor modifications in the proofs.
The impossibility results show where this analogy breaks down: strong EF1 allocations fail to exist in an instance whose eligibility structure is based on the biplane of order $2$. %

\subsection{Related Work}\label{sec:related-work}

\paragraph{Fair division of indivisible items.}
The classical fairness benchmarks are \emph{envy-freeness}~\cite{F67resource,V74equity}, requiring that no agent prefer another's bundle, and \emph{proportionality}~\cite{Steinhaus48}, guaranteeing each agent an equal share of the total value; see~\cite{AAB+23fair} for a survey. Since neither is always attainable with indivisible items, relaxations are studied: EF1~\cite{LMMS04,Budish11} asks that envy be eliminable by removing one item from the envied bundle, EFX~\cite{CKM+19} requires this for every item, and $\PropOne$~\cite{CFS17} that a proportional share be reached by adding one item. These notions extend from goods to chores~\cite{ACIW22}. Offline, EF1 is achievable for both goods~\cite{LMMS04} and chores~\cite{ACIW22,BSV21}, whereas EFX allocations are known to exist only in special cases~\cite{PR20,CGM24} and may fail to exist in general, both for goods~\cite{AMMSW26,MS26} and for chores~\cite{CS24chores}.

A related line of work studies fair division under combinatorial constraints; see~\cite{S21constraints} for a survey. Several works investigate the fair allocation subject to pairwise \emph{conflict constraints}, where two conflicting items cannot be assigned to the same agent~\cite{LLZ21fair, HH22fair, KEG+24fair, IMY25dividing, EGI+26fair, MarkakisS2026matchings, Haviv2026colorability}. These works study the existence of allocations that are simultaneously approximately envy-free and economically efficient (e.g., maximal or Pareto optimal), in both the goods and~chores~settings. Matroid constraints on bundles~\cite{CooksonS2026matroid} and budget constraints in the online setting~\cite{CohenTGW2026budget} restrict what an agent may hold; in all of these models, however, every item is eligible for every agent, and the constraint acts on bundles rather than on agent--item pairs.

Closest to our model are two lines of work in which items are tied to specific agents. Under \emph{restricted additive} valuations~\cite{ARS22ef2x,KSS24tale}, every good $j$ has a value $p_j$ and each agent values it at either $p_j$ or $0$; eligibility is thus encoded in the valuations, but a good may still be given to an agent who values it at zero. Known guarantees are EF2X and $(\sqrt{2}/2)$-EFX for goods~\cite{ARS22ef2x,KSS24tale} and EFX together with MMS for chores~\cite{LWZ26chores}. In the \emph{graphical} model of Christodoulou et al.~\cite{CFKS23graphs}, agents are vertices and goods are edges valued only by their endpoints; an \emph{orientation} gives each edge to one of its endpoints and is exactly a restricted assignment with $|\M(j)|=2$, with multigraphs and hypergraphs covering larger eligibility structures. EFX orientations need not exist and are NP-hard to recognize~\cite{CFKS23graphs}, with a characterization of graphs that always admit them in~\cite{ZM25structure}, whereas EF1 orientations always exist for monotone valuations and arbitrary eligibility sets~\cite{DEGK25orientations}; further work covers chores and mixed items~\cite{ZWLL24landscape}, EFX allocations on multigraphs~\cite{AAMMS26multigraphs}, and MMS~\cite{CM26mms}. Most relevant to us, Sun, Wang, and Li~\cite{SWL26orientations} recently and independently proposed, under the name \emph{relevance constraints} and for general additive valuations, the same eligibility-aware proportional share that we use. They show offline that proportional orientations may fail to exist and are NP-hard to recognize, that $\PropOne$ orientations always exist and can be computed in polynomial time, and that $\PropX$ orientations may fail to exist; for identical valuations, our offline results in \Cref{sec:offline} are special cases of theirs. All of these works are offline. Our focus is the online problem, for which no prior results exist.

In online fair division, items arrive sequentially and must be allocated irrevocably~\cite{AAG+15online,AW19strategy,BKP+24fair,HHI+24class,KMN+26online,HPP+19achieving}; see~\cite{AW20online} for a survey. The unrestricted picture for $\PropOne$ is now well understood. Under identical valuations, $\PropOne$ is achievable online for both goods and chores~\cite{ChenT2026competitive}. Under general additive valuations, exact $\PropOne$ for goods is achievable when each agent's total value is known in advance~\cite{neoh2026online}; without this information, natural greedy rules give no bounded guarantee~\cite{choo2025approximateproportionalityonlinefair}, no online algorithm achieves a constant factor against an adaptive adversary~\cite{NeohT2026closing}, and the best deterministic guarantee is $O(\log(mn))$-$\PropOne$, with a matching lower bound growing in the number of items~\cite{DuS2026online}. Proportionality under predictions of each agent's total or maximum value has also been studied~\cite{banerjee2022proportionallyfaironlineallocation,choo2025approximateproportionalityonlinefair,NeohT2026closing}, as have online MMS guarantees for chores~\cite{SongTWZ2025mms} and online EF1 with a small buffer of reassignable items~\cite{AmanatidisGMP2026buffers}. None of these works allows eligibility constraints; under identical valuations, they give exact online $\PropOne$, so the $\Theta(\sqrt m)$ and $\Theta(\log m)$ losses we establish are due to eligibility alone. Two related temporal variants relax the online model in different directions: in \emph{temporal} fair division, the arrival order is known in advance and approximate envy-freeness must hold at every step~\cite{ELL+25temporal,CES25temporal,ChoiL2026temporal,ChoiLT2026manna}, possibly with assignments deferred within a fixed window~\cite{CL26temporal}; in \emph{repeated} fair division, the same set of items recurs over time~\cite{ILN+24repeated,CN24repeatedly}. Fair division has also been combined with scheduling structure in the form of interval scheduling with machines as agents~\cite{BorstSV2026interval}, where the constraint is non-overlap on a machine rather than eligibility.

\paragraph{Restricted assignment in scheduling.}
Restricted assignment is a central model in scheduling and one of the most prominent eligibility-constrained special cases of unrelated-machine scheduling. In the chore interpretation, the global objective is makespan minimization: Lenstra, Shmoys, and Tardos~\cite{lenstra1990approximation} gave a $2$-approximation for unrelated machines and showed that no factor below $3/2$ is possible unless $\mathrm{P}=\mathrm{NP}$; the best known polynomial-time guarantee is $2-1/m$~\cite{ShchepinV05}, even for restricted assignment, with better bounds known in quasi-polynomial time and for special cases~\cite{jansen2020quasi,chakrabarty20141,annamalai2019lazy,DBLP:conf/soda/BamasLMRS24}. In the goods interpretation, the corresponding global objective is max-min fair allocation, or Santa Claus~\cite{bansal2006santa}. While the general unrelated version admits an $n^\varepsilon$-approximation for every constant $\varepsilon>0$~\cite{chakrabarty2009allocating}, constant-factor approximations are known for restricted assignment~\cite{bansal2006santa,feige2008allocations,DBLP:journals/talg/AsadpourFS12,annamalai2017combinatorial,HaxellS2023improved}. Recent work relates the chore and goods viewpoints for global scheduling objectives~\cite{DBLP:conf/soda/BamasLMRS24}, although the reductions fail in the restricted-assignment setting; our equivalence for envy-based fairness is analogous, but applies to local eligibility-aware fairness notions.

Online variants of these global objectives behave differently: makespan minimization admits a tight $\mathcal{O}(\log m)$ bound even for restricted assignment~\cite{DBLP:journals/jacm/AspnesAFPW97}, whereas adversarial online max-min allocation is already pessimistic for identical valuations~\cite{DBLP:conf/esa/AzarE97}. Although the feasibility model is the same, our objectives are local rather than global: each machine is compared to its own eligibility-aware benchmark, or to other machines only through jobs relevant under eligibility constraints. Thus offline LP-rounding machinery remains useful, but both the objectives and the standard methods of online load balancing  differ from what is needed for our local fairness guarantees.

\section{Preliminaries}
\label{sec:preliminaries}
We phrase the problem in the language of restricted assignment: machines correspond to agents and jobs to items, and we use the two vocabularies interchangeably. We consider the problem of allocating $n$ jobs among a set $\M$ of $m$ machines. Each job $j \in [n] \coloneqq \{1,2,\ldots,n\}$ has a size $p_j > 0$ and an eligibility set $\M(j) \subseteq \M$ of machines that may receive it. We write $k_j \coloneqq |\M(j)|$. An allocation $\sigma$ assigns each job $j$ to some machine $\sigma(j) \in \M(j)$, and we write
\[
L_i(\sigma) \coloneqq \sum_{j\in [n] \,:\, \sigma(j) = i} p_j
\]
for the total size assigned to machine $i\in \M$.

We study two readings of the model. In the \emph{chore} setting, jobs are burdens: $p_j$ is a disutility, or processing time, $L_i(\sigma)$ is the load of machine $i$, and each machine wants its load to be small. In the \emph{goods} setting, jobs are desirable: $p_j$ is a utility or value, $L_i(\sigma)$ is the total utility of machine $i$, and each machine wants it to be large. %

\paragraph{Proportional share.}
We define the \emph{proportional share} of machine $i$ as
\[
\PS_i \coloneqq \sum_{j \in [n] \,:\, i \in \M(j)} \frac{p_j}{k_j},
\]
the amount machine $i$ receives in the natural fractional relaxation in which every job is split evenly among its eligible machines. The share accounts for eligibility directly: $\PS_i$ depends only on the jobs machine $i$ can actually receive. Since each job's size is fully distributed among its eligible machines,
\[
\sum_{i \in \M} \PS_i = \sum_{j\in[n]} p_j,
\]
so the shares partition the total size exactly. Throughout the paper, $\PS_i$ serves as our eligibility-aware benchmark: an upper bound to stay below in the chore setting, and a lower bound to reach in the goods setting.

\paragraph{Proportionality.}
An allocation $\sigma$ is \emph{proportional} ($\Prop$) if every machine meets its share: $L_i(\sigma) \le \PS_i$ for all $i \in \M$ in the chore setting, and $L_i(\sigma) \ge \PS_i$ for all $i \in \M$ in the goods setting. Because exact proportionality is generally unattainable with indivisible jobs, we work with two standard relaxations, whose form depends on the setting. For chores, the relaxations discount a job the machine \emph{received}; for goods, they credit a job the machine is eligible for but \emph{did not receive}.

In the chore setting, $\sigma$ satisfies \emph{proportionality up to one job} ($\PropOne$) if, for every machine $i$,
\[
	L_i(\sigma) - p_{\max}(i, \sigma) \, \le\  \PS_i,
\]
where $p_{\max}(i, \sigma) \coloneqq \max_{j \, : \, \sigma(j)=i}p_j$ denotes the size of the largest job allocated to $i$ or $0$ if $i$ does not receive a job.
It satisfies the stronger \emph{proportionality up to any job} ($\PropX$) if $L_i(\sigma)-p_j\le \PS_i$ for every $i\in\M$ and every job $j$ with $\sigma(j)=i$.

In the goods setting, $\sigma$ satisfies $\PropOne$ if, for every machine $i\in\M$, either $L_i(\sigma)\ge \PS_i$, or there is a job $j$ with $i\in\M(j)$ and $\sigma(j)\ne i$ such that
\[
	L_i(\sigma)+p_j \, \ge\,  \PS_i.
\]
It satisfies $\PropX$ if this holds for every such job $j$, that is, if $L_i(\sigma)+p_j\ge \PS_i$ for every $i\in\M$ and every job $j$ with $i\in\M(j)$ and $\sigma(j)\ne i$.
Note that only jobs machine $i$ is eligible for may be credited; allowing arbitrary jobs would defeat the purpose of an eligibility-aware notion. In both settings, $\Prop$ implies $\PropX$, and $\PropX$ implies $\PropOne$. We defer the eligibility-aware envy notions to \Cref{sec:envy}, where they are introduced together with the corresponding results.

\paragraph{Approximate proportionality.}
Our online results are stated in terms of an approximate version of $\PropOne$, similar to the notion of \cite{choo2025approximateproportionalityonlinefair}.
Let $\alpha \ge 1$. In the chore setting, an allocation $\sigma$ satisfies \emph{$\alpha$-$\PropOne$} if for every machine $i \in \M$ that receives at least one job,
\[
L_i(\sigma) - p_{\max}(i, \sigma) \; \le \; \alpha \cdot \PS_i.
\]
In the goods setting, $\sigma$ satisfies $\alpha$-$\PropOne$ if for every machine $i \in \M$, either $L_i(\sigma) \ge \tfrac{1}{\alpha}\,\PS_i$, or there is a job $j$ with $i \in \M(j)$ and $\sigma(j) \ne i$ such that
\[
L_i(\sigma) + p_{j} \; \ge \; \frac{1}{\alpha} \cdot \PS_i.
\]
We may omit $\sigma$ from our notation whenever the allocation is clear from the context.

\subsection{Offline allocation of chores and goods}\label{sec:offline}
We first consider the offline problem, where the entire instance is known in advance. The results here are complete: in both settings, $\PropOne$ is the strongest proportionality guarantee that survives eligibility, and it is achievable in polynomial time. These facts were obtained independently, for general additive valuations and mixed items, by Sun, Wang, and Li~\cite{SWL26orientations}; we include short proofs via LP rounding for completeness. We defer detailed proofs to \Cref{app:offline}.

Although the goods and chore settings share their notation, their fairness notions do not coincide: an allocation may be fair in one reading and unfair in the other, in either direction.

\begin{restatable}{observation}{ObsPropOneGoodsNeqChores}
\label{obs:propone-goods-neq-chores}
There is an instance of restricted assignment and an allocation that satisfies \PropOne{} in the chore setting but not in the goods setting. Moreover, there is an instance and an allocation that satisfies \PropOne{} in the goods setting but not in the chore setting.
\end{restatable}

The two settings therefore require separate treatment. Already offline, eligibility rules out the stronger proportionality notions.

\begin{restatable}{observation}{ObsOfflineNoPropX}\label{obs:offline-no-propx}
In restricted-assignment chore and goods allocation, \PropX{} is not always achievable. Consequently, exact proportionality is not always achievable either.
\end{restatable}

In contrast, $\PropOne$ remains achievable.
The proof for the chore setting is a direct application of the rounding argument of Lenstra, Shmoys, and Tardos~\cite{lenstra1990approximation}.
After normalizing each eligible entry by the proportional share of its machine, the fractional allocation $x_{ij}=1/k_j$ satisfies the assignment LP with unit machine bounds.
Rounding yields an integral allocation in which~each machine violates its normalized bound by at most one assigned job in the chore interpretation; rescaling gives $\PropOne$.
The goods case is obtained analogously using the rounding theorem of Bez{\'a}kov{\'a} and Dani~\cite{bezakova2005allocating} based on the work of Lenstra, Shmoys and Tardos \cite{lenstra1990approximation}. See \Cref{app:offline}.

\begin{restatable}{theorem}{ThmOfflinePropOne}\label{thm:offline-propone}
For every instance of restricted-assignment chore and goods allocation, there exists a \PropOne{} allocation, and such an allocation can be computed in polynomial time.
\end{restatable}

\section{Online Proportionality for Chores}
\label{sec:chore_online}
In the online model, jobs arrive sequentially and must be assigned irrevocably upon arrival. In contrast to the offline setting, the algorithm no longer has access to the full instance in advance and must make decisions without knowing future arrivals. We prove that the optimal approximation factor for \PropOne{} in restricted-assignment chore allocation is $\Theta(\sqrt m)$.

We start with unit-size jobs, where the difficulty comes solely from the eligibility constraints. We give a simple greedy algorithm based on $\alpha$-slack and prove that it achieves $O(\sqrt m)$-\PropOne{}. We complement this with a matching $\Omega(\sqrt m)$ lower bound for deterministic algorithms. Finally, we extend the upper bound to arbitrary processing times by replacing the simple greedy rule with a Fit/Shadow algorithm that can handle large arriving jobs.

\subsection{Unit Processing Times}
\label{sec:chores-unit}

We first consider unit-size jobs, where all difficulty comes from eligibility. A tempting rule is to assign each arriving job to an eligible machine with maximum current $\PropOne$-slack. This direct analogue of least-loaded machine balancing is too myopic and can fail to achieve an $o(m)$-approximation; see \Cref{app:greedy-fails}. We therefore use a slightly more conservative slack notion. For a parameter $\alpha\ge 1$, the $\alpha$-slack of machine $i$ is
\[
s_\alpha(i) \coloneqq \alpha\cdot \PS_i-L_i(\sigma).
\]
For unit-size jobs, controlling this slack before the last large-eligibility job is assigned to a machine is enough to recover $\alpha$-\PropOne{} after removing that job.

\paragraph{Greedy \boldmath$\alpha$-slack algorithm.}
Fix $\alpha=\sqrt m+1$. When a job $j$ arrives, assign it to an eligible machine with maximum $\alpha$-slack, breaking ties arbitrarily. The proportional share used in the slack includes the arriving job $j$, but not future jobs.

\begin{theorem} \label{thm:chore_online_unit_size_greedy_ratio}
Let $\alpha=\sqrt m+1$. In restricted-assignment chore allocation for jobs with unit processing times, Greedy $\alpha$-slack assignment yields an $\alpha$-\PropOne{} allocation.
\end{theorem}

For the analysis, we use a potential function $\Phi$ which penalizes machines with small positive $\alpha$-slack exponentially, and gives an additional unit penalty to machines with negative $\alpha$-slack.

As the load from a job that has at most $\alpha$ eligible machines is covered by its own proportional share amplified by $\alpha$, the problematic jobs are those that have more than $\alpha$ eligible machines.
The machine $i$ receiving such a job $j$ incurs a full unit of load while its proportional share only increases by the small fraction $\alpha/k_j$, resulting in a net decrease of machine $i$'s $\alpha$-slack and an increase in $\Phi$.
However, when the job is allocated according to the greedy criterion, the $\alpha$-slacks of the other $k_j - 1$ eligible machines are not much larger than that of $i$.
Hence, the increase in proportional share on these other machines yields a substantial decrease of $\Phi$ that fully negates the increase, implying that $\Phi$ never becomes larger than its initial value.

Since machines with negative $\alpha$-slack incur a large discrete penalty, it follows that there cannot be $\sqrt{m}$ machines that have negative $\alpha$-slack at the same time.

To construct our potential function $\Phi \coloneqq \sum_{i \in \M} \Phi_i$, we define the potential term of a machine~$i$~as
\[
    \Phi_i \coloneqq \mathds{1}_{s_\alpha(i) < 0} + \mathds{1}_{s_\alpha(i) < 0} \cdot \sqrt{m}^{-1} \cdot \vert s_\alpha(i) \vert + \sum_{\hat{s} \in \bZ_{\geq 0}: s_\alpha(i) < \hat{s} + 1} \min \{1, \hat{s} + 1 - s_\alpha(i)\} \cdot \sqrt{m}^{-(\hat{s} + 1)}
\]
which can equivalently be expressed as
\begin{align*}
    \Phi_i = \mathds{1}_{s_\alpha(i) < 0} + \int_{s_\alpha(i)}^{\infty} \sqrt{m}^{-\rho(t)} dt
    \qquad \text{with} \qquad
    \rho(t) = \begin{cases}
        1 & t \leq 0\\
        \lceil t \rceil & t > 0.
    \end{cases}
\end{align*}

\Cref{fig:chore_online_unit_size_plot_phi_i} shows a plot of the term $\Phi_i$ as a function of the $\alpha$-slack $s_\alpha(i)$.
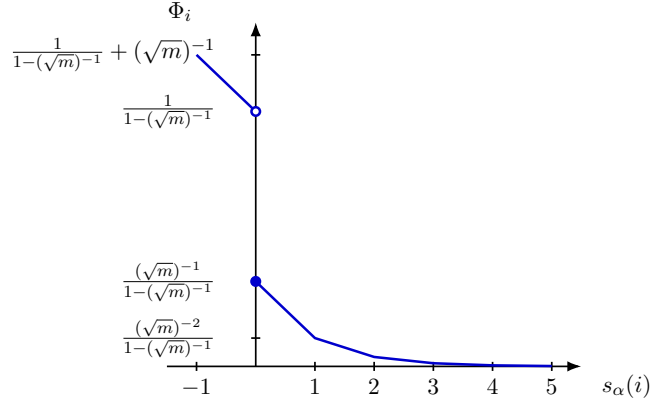
\begin{figure}[!t]
    \centering
    \begin{tikzpicture}[scale=0.8]
        \colorlet{blue}{blue!80!black}
        \begin{axis}[
            axis lines = middle,
            axis line style={-{Latex}, thick},
            tick style={black, thick},
			major tick style={black, thick},
			minor tick style={black, thick},
			tick label style={black},
            xlabel = {$s_\alpha(i)$},
            ylabel = {$\Phi_i$},
            ylabel style={at={(axis description cs:0.08,0.98)}, anchor=south east},
            xlabel style={at={(axis description cs:1.03,0)}, anchor=north west},
            xmin=-1.5, xmax=5.5,
            ymin=0, ymax=1.35,
            xtick={-1,0,1,2,3,4,5},
            ytick={0,1/9,1/3,1,11/9},
            yticklabels={
                0,
                $\frac{(\sqrt{m})^{-2}}{1 - (\sqrt{m})^{-1}}$,
                $\frac{(\sqrt{m})^{-1}}{1 - (\sqrt{m})^{-1}}$,
                $\frac{1}{1 - (\sqrt{m})^{-1}}$,
                $\frac{1}{1 - (\sqrt{m})^{-1}} + (\sqrt{m})^{-1}$},
            yticklabel style={xshift=-1.2em},
            clip=false,
            every axis plot/.append style={very thick},
        ]
        \addplot [
            domain=-1:0,
            color=blue,
            ]
            coordinates {
            (-1,11/9)(0,1)
            };
        \addplot [ %
            only marks,
            color=blue,
            mark=*,
            mark options={fill=white}
            ]
            coordinates {(0,1)
            };
        \addplot [ %
            only marks,
            color=blue,
            mark=*
            ]
            coordinates {(0,1/3)
            };
        \addplot [
            domain=0:5,
            color=blue,
            ]
            coordinates {
            (0,1/3)(1,1/9)(2,1/27)(3,1/81)(4,1/243)(5,1/729)
            };
        \end{axis}
    \end{tikzpicture}
    \caption{Plot of the potential term $\Phi_i$ as a function of the $\alpha$-slack of machine $i$.}
    \label{fig:chore_online_unit_size_plot_phi_i}
\end{figure}

\begin{lemma} \label{lem:chore_online_unit_size_few_machines_have_negative_s_alpha}
    If $\Phi$ is at most its initial value,
    then less than $\sqrt{m}$ machines have negative $\alpha$-slack.
\end{lemma}

\begin{proof}
    Before any jobs arrive, the value of $\Phi$ is
    \[
        \Phi = m \cdot \sum_{\hat{s} = 0}^{\infty} \sqrt{m}^{-(\hat{s} + 1)} %
        = \frac{\sqrt{m}}{1 - \sqrt{m}^{-1}}.
    \]
    Now consider any configuration in which $\smash{\Phi \leq \frac{\sqrt{m}}{1 - \sqrt{m}^{-1}}}$. The contribution of every machine $i$ with negative $\alpha$-slack to $\Phi$ is greater than
    \[
        \Phi_i > 1 + \sum_{\hat{s} \in \bZ_{\geq 0}} \sqrt{m}^{-(\hat{s} + 1)} %
        = 1 + \sum_{\hat{s} = 1}^\infty \sqrt{m}^{-\hat{s}}%
        = \frac{1}{1 - \sqrt{m}^{-1}}.
    \]
    As the contribution of any other machine to $\Phi$ is nonnegative, there have to be strictly fewer than $\sqrt{m}$ machines with negative $\alpha$-slack for $\smash{\Phi \leq \frac{\sqrt{m}}{1 - \sqrt{m}^{-1}}}$ to hold.
\end{proof}

\begin{lemma} \label{lem:chore_online_unit_size_potential_does_not_increase_in_greedy}
    $\Phi$ does not increase when a new job arrives and is assigned to the eligible machine with maximum $\alpha$-slack.
\end{lemma}
\begin{proof}
    Let $j$ be the new job. If $k_j \leq \alpha$, then~$j$ does not cause any $\alpha$-slack to decrease and we are done. In the following, we therefore assume that $k_j > \alpha$.

    For every machine $i$, let $s_\alpha(i)$ be the $\alpha$-slack before the new job arrives and let $s_\alpha^{\text{new}}(i)$ be the $\alpha$-slack when $j$ is taken into account. Let $i^* \in \M(j)$ be the machine that $j$ is allocated to. In particular, $i^*$ is a machine with maximum $\alpha$-slack among the machines in $\M(j)$. Since the $\alpha$-slack and the load only change for machines in $\M(j)$, the difference between the new and the old potential is
    \[
        \Phi^{\text{new}} - \Phi^{\text{old}} = \sum_{i \in \M(j)} \Phi^{\text{new}}_i - \Phi^{\text{old}}_i.
    \]
    For $i \in \M(j)$, let
    \[
        \widehat{s}_\alpha(i) \coloneqq \begin{cases}
        0 & s_\alpha(i) < 0\\
        \left\lfloor s_\alpha(i) + \frac{\alpha}{k_j} \right\rfloor & s_\alpha(i) \geq 0
        \end{cases}.
    \]
    Note that $i^*$ is a machine with maximum $\widehat{s}_\alpha(i)$ among the machines in $\M(j)$.

    Now, for any machine $i \neq i^*$, the $\alpha$-slack of machine $i$ increases by $\smash{\frac{\alpha}{k_j}}$. By the construction of $\Phi_i$, we have
    \[
        \Phi^{\text{new}}_i - \Phi^{\text{old}}_i \leq -\frac{\alpha}{k_j} \cdot \sqrt{m}^{-(\widehat{s}_\alpha(i) + 1)} %
        \leq -\frac{\alpha}{k_j} \cdot \sqrt{m}^{-(\widehat{s}_\alpha(i^*) + 1)}.
    \]
    For $i^*$, the $\alpha$-slack decreases by $\frac{k_j - \alpha}{k_j}$. The new $\alpha$-slack of machine $i^*$ is
    \[
        s_\alpha^{\text{new}}(i^*) = s_\alpha(i^*) - \frac{k_j - \alpha}{k_j}.
    \]
    In particular, it holds that
    \[
        \left\lfloor s_\alpha^{\text{new}}(i^*) \right\rfloor = \left\lfloor s_\alpha(i^*) - \frac{k_j - \alpha}{k_j}\right\rfloor
        = \widehat{s}_\alpha(i^*) - 1.
    \]

    Let $\smash{\Delta = s_\alpha(i^*) - \widehat{s}_\alpha(i^*) \leq \frac{k_j - \alpha}{k_j}}$. The new $\alpha$-slack of machine $i^*$ can be bounded by analyzing the difference in potential:
    \begin{align*}
        \Phi^{\text{new}}_{i^*} - \Phi^{\text{old}}_{i^*} &\leq \begin{cases}
            \frac{k_j - \alpha}{k_j} \cdot \sqrt{m}^{-\widehat{s}_\alpha(i^*)} & \text{if } \widehat{s}_\alpha(i^*) \geq s_\alpha(i^*),\\[1ex]
            \Delta \cdot \sqrt{m}^{-(\widehat{s}_\alpha(i^*) + 1)} + \left(\frac{k_j - \alpha}{k_j} - \Delta\right) \cdot \sqrt{m}^{-\widehat{s}_\alpha(i^*)} & \text{if } s_\alpha(i^*) > \widehat{s}_\alpha(i^*) > 0,\\[1ex]
            \Delta \cdot \sqrt{m}^{-(\widehat{s}_\alpha(i^*) + 1)} + 1 & \text{if } s_\alpha(i^*) > \widehat{s}_\alpha(i^*) = 0.
        \end{cases}\\[2ex]
        &\leq \begin{cases}
            \sqrt{m}^{-\widehat{s}_\alpha(i^*)} & \text{if } \widehat{s}_\alpha(i^*) \geq s_\alpha(i^*),\\[1ex]
            \frac{k_j - \alpha}{k_j} \cdot \sqrt{m}^{-(\widehat{s}_\alpha(i^*) + 1)} + \sqrt{m}^{-\widehat{s}_\alpha(i^*)} & \text{if } s_\alpha(i^*) > \widehat{s}_\alpha(i^*) > 0,\\[1ex]
            \frac{k_j - \alpha}{k_j} \cdot \sqrt{m}^{-(\widehat{s}_\alpha(i^*) + 1)} + \sqrt{m}^{-\widehat{s}_\alpha(i^*)} & \text{if } s_\alpha(i^*) > \widehat{s}_\alpha(i^*) = 0.
        \end{cases}\\[2ex]
        &\leq \frac{k_j - \alpha}{k_j} \cdot \sqrt{m}^{-(\widehat{s}_\alpha(i^*) + 1)} + \sqrt{m}^{-\widehat{s}_\alpha(i^*)}\\[1ex]
        &= \left(\sqrt{m} + \frac{k_j - \alpha}{k_j}\right) \cdot \sqrt{m}^{-(\widehat{s}_\alpha(i^*) + 1)}.
    \end{align*}

    The total potential difference is thus at most
    \begin{align*}
        \Phi^{\text{new}} - \Phi^{\text{old}} &= \Phi^{\text{new}}_{i^*} - \Phi^{\text{old}}_{i^*} + \sum_{i \in \M(j), i \neq i^*} \Phi^{\text{new}}_i - \Phi^{\text{old}}_i\\
        &\leq \left(\sqrt{m} + \frac{k_j - \alpha}{k_j}\right) \cdot \sqrt{m}^{-(\widehat{s}_\alpha(i^*) + 1)} + (k_j - 1) \cdot \left(-\frac{\alpha}{k_j} \cdot \sqrt{m}^{-(\widehat{s}_\alpha(i^*) + 1)}\right)\\
        &= \sqrt{m}^{-(\widehat{s}_\alpha(i^*) + 1)} \cdot \left(\sqrt{m} + \frac{k_j - \alpha}{k_j} - (k_j - 1) \cdot \frac{\alpha}{k_j}\right)\\
        &= \sqrt{m}^{-(\widehat{s}_\alpha(i^*) + 1)} \cdot (\sqrt{m} + 1 - \alpha)\\
        &= 0
    \end{align*}
    where we used that $\alpha = \sqrt{m} + 1$. This shows that the potential never increases.
\end{proof}
Finally, we have the tools to prove \cref{thm:chore_online_unit_size_greedy_ratio}.

\begin{proof}[Proof of \cref{thm:chore_online_unit_size_greedy_ratio}]
By \Cref{lem:chore_online_unit_size_potential_does_not_increase_in_greedy}, the greedy algorithm maintains the invariant $\Phi \leq \Phi^0$, where $\Phi^0$ is the initial value of the potential.
Consider any machine $i$ that receives at least one job, and let $j$ be the
last job assigned to $i$ that has more than $\alpha$ eligible machines.
Immediately before assigning $j$, fewer than $\sqrt{m}<\alpha<k_j$ machines have negative $\alpha$-slack by \Cref{lem:chore_online_unit_size_few_machines_have_negative_s_alpha}. Hence, since the algorithm chooses a machine with maximum $\alpha$-slack among the eligible machines, machine $i$ has nonnegative $\alpha$-slack before receiving~$j$, that is, $L_i\le \alpha\cdot\PS_i$, where $L_i$ and $\PS_i$ are computed just before $j$ arrives.
Removing the processing time of $j$ restores this inequality as
the proportional share $\PS_i$ can never decrease. Hence $L_i - 1 \leq \alpha \cdot \PS_i$, and
since all jobs have unit processing time, this is exactly $\alpha$-\PropOne{}.

Any jobs allocated to $i$ after $j$ have at most $\alpha$ eligible machines and
therefore increase the right hand side by at least as much as the left hand side, preserving the inequality.
Any other jobs arriving after $j$ can only increase the proportional share of machine $i$ and do not increase its load.
Hence, the $\alpha$-\PropOne{} guarantee still holds for the final allocation.
\end{proof}

To complement this upper bound, we show that no deterministic online algorithm can achieve $o(\sqrt m)$-\PropOne{}, even for unit-size jobs. The adversary exploits the ability to target precisely those machines that already have comparatively high load.

\begin{restatable}{theorem}{ThmChoreOnlineLB}
\label{thm:chore-online-LB}
In restricted-assignment chore allocation, no deterministic online algorithm can guarantee $\alpha$-\PropOne{} with $\alpha=o(\sqrt m)$, even for unit-size jobs. This bound holds even with prior knowledge of all final proportional shares.
\end{restatable}
\begin{proof}
Assume for now that $\sqrt{m}$ is an integer. We will get rid of this assumption later.

Fix a deterministic online algorithm. Consider the following adversarial instance with unit-length jobs.
First, the adversary sends $\sqrt{m}$ jobs for which all $m$ machines are eligible.
After these $\sqrt{m}$ jobs, if some machine received at least two of them, then after
removing one job its remaining load is at least $1$, while its proportional share is
$\sqrt{m}/m=1/\sqrt{m}$.
Thus the allocation already violates
$\alpha$-\PropOne{} for every $\alpha<\sqrt{m}$.
Hence, we may assume that these jobs are assigned to distinct machines.

Let $S$ be the set of machines that received one of these first $\sqrt{m}$ jobs.
The adversary now releases one final job $j^*$ with $\M(j^*) = S$.
Let $i \in S$ be the machine receiving $j^*$.
Then, after removing one job from $i$, the remaining load is $L_i-1=1$.
On the other hand,
\[
    \PS_i = \frac{\sqrt{m}}{m}+\frac{1}{\sqrt{m}}
    = \frac{2}{\sqrt{m}}.
\]
Thus, $\alpha$-\PropOne{} requires
\[
    1 \leq \alpha \cdot \PS_i =  \frac{2\alpha}{\sqrt{m}},
\]
and therefore $\alpha \geq \frac{\sqrt{m}}{2}$.

If the algorithm knows the final proportional shares of all machines, the adversary can release an additional set of $\sqrt{m} - 1$ unit-size jobs for which exactly the machines \emph{not} in $S$ are eligible.
After adding these jobs, the proportional share of every machine is equal to ${2}/{\sqrt{m}}$.
Hence, the algorithm cannot identify the machines in $S$ before the $(\sqrt{m} + 1)$-th job arrives.

To get rid of our assumption that $\sqrt{m}$ is an integer, the adversary can run the above instance on $(\lfloor \sqrt{m} \rfloor)^2$ machines and ignore the remaining ones, achieving a lower bound of $\Omega(\lfloor \sqrt{m} \rfloor) = \Omega(\sqrt{m})$.
\end{proof}

\subsection{General Processing Times}
\label{sec:chore-online-general}

We now extend the $\Theta(\sqrt m)$ upper bound to arbitrary processing times. The unit-size greedy rule no longer suffices: a single large job can destroy the $\alpha$-slack of the machine receiving it. Our algorithm is still mostly greedy based on the slacks of the eligible machines towards violating $\alpha$-\PropOne{} for $\alpha = 8 \sqrt{m}$. We call this quantity the $\alpha$-\PropOne-slack of machine~$i$ which is given by
\[
    s_{\alpha\PropOne}(i) \coloneqq \alpha \cdot \PS_i - L_i(\sigma) + p_{\max}(i, \sigma).
\]

\begin{restatable}{theorem}{ThmChoreOnlineGeneral}
\label{thm:general_upper_bound}
There is a deterministic online algorithm for restricted-assignment chore allocation with general processing times that achieves $O(\sqrt{m})$-\PropOne{}.
\end{restatable}

\paragraph{Fit/Shadow algorithm.}
When a job $j^*$ arrives with $k_{j^*}\le\alpha$, assign it to an arbitrary eligible machine and exclude it from all subsequent slack and $p_{\max}$ computations.
For a job $j^*$ with $k_{j^*}>\alpha$,
we determine the hypothetical $\alpha$-\PropOne-slack for all eligible machines \emph{after} adding $\frac{\alpha \cdot p_{j^*}}{k_{j^*}}$ to their proportional shares. We then find the eligible machine~$i$ for which this updated $s_{\alpha\PropOne}(i)$ is largest. We distinguish two cases:
\begin{enumerate}
    \item \textbf{Case 1 (Fit):} If the largest updated slack satisfies $s_{\alpha\PropOne}(i) \geq 2p_{j^*}$, assign $j^*$ to machine~$i$. We say that $j^*$ \emph{fits} on $i$.

    \item \textbf{Case 2 (Shadow):} Otherwise, consider the eligible machines with $p_{\max}(i, \sigma) \leq \alpha p_{j^*}/k_{j^*}$.
    Among these machines, assign $j^*$ to one with maximum updated $\alpha$-\PropOne{}-slack.
    We say that~$j^*$ \emph{shadows} the previous largest job on that machine.
\end{enumerate}

We will now give a proof sketch for \Cref{thm:general_upper_bound}. The full proof is deferred to \Cref{app:general-upper-bound-online-chores}.

\begin{proof}[Proof of \Cref{thm:general_upper_bound} (sketch)]

    Jobs with $k_{j^*}\le\alpha$ cannot decrease any machine's $\alpha$-\PropOne-slack, so it suffices to analyze the virtual instance that excludes them.

    Initially, all slacks are zero. Let $j^*$ be a job with $k_{j^*} > \alpha$.

    We first observe that in both the fit and the shadow case, the $\alpha$-\PropOne-slack of the machine $i$ to which $j^*$ is allocated remains nonnegative: in the fit case, the fit condition implies that $i$ has more than enough $\alpha$-\PropOne-slack to absorb the new job. In the shadow case, we use the fact that $p_{j^*}$ may be subtracted in the \PropOne{} computation instead of the previous $p_{\max}(i, \sigma)$. This prevents the $\alpha$-\PropOne-slack from decreasing by more than $p_{\max}(i, \sigma)$ due to the load of $j^*$. By the shadow condition, the increase in scaled proportional share caused by $j^*$ is large enough to compensate for this decrease.

    It remains to show that the algorithm never gets stuck: whenever no eligible machine can fit a job, at least one eligible machine must satisfy the shadow condition. We will prove the stronger invariant that every job that does not fit can shadow the largest jobs on more than half of its eligible machines.

    To prove this invariant, we use induction over the arriving jobs. The induction step consists of proving the stronger invariant for a job $j^*$ under the assumption that it held for all previous jobs. As this assumption is empty for the first job, no additional induction start is needed.

    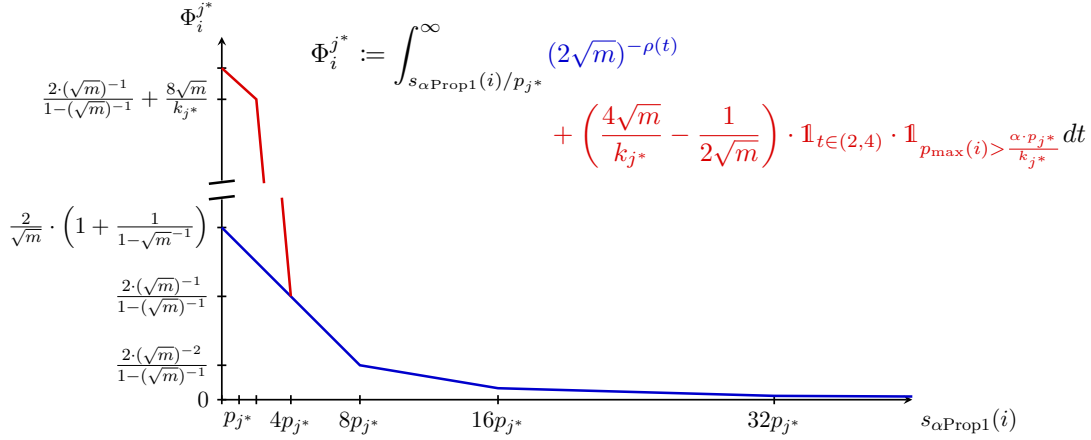
\begin{figure}[tbh]
        \centering
        \begin{tikzpicture}[baseline,scale=0.8]
            \colorlet{blue}{blue!80!black}
            \colorlet{red}{red!85!black}
            \node[anchor=west] at (1.3, 1.6) {\small $\displaystyle \begin{aligned}
                \Phi^{j^*}_i \coloneqq \int_{s_{\alpha\PropOne}(i) / p_{j^*}}^\infty &\textcolor{blue}{(2 \sqrt{m})^{- \rho(t)}}\\
                &\textcolor{red}{{}+ \left(\frac{4 \sqrt{m}}{k_{j^*}} - \frac{1}{2\sqrt{m}}\right) \cdot \mathds{1}_{t \in (2, 4)} \cdot \mathds{1}_{p_{\max}(i) > \frac{\alpha \cdot p_{j^*}}{k_{j^*}}}} dt \end{aligned}$};
            \begin{groupplot}[
                group style={group size=1 by 2, vertical sep=0pt},
                width=13cm,
                axis line style={-{Latex}, thick},
                tick style={black, thick},
                major tick style={black, thick},
                minor tick style={black, thick},
                tick label style={black},
                xlabel = {$s_{\alpha \PropOne}(i)$},
                xlabel style={at={(axis description cs:1,.05)}, anchor=west},
                xmin=0, xmax=40,
                xtick={0,1,2,4,8,16,32},
                xticklabels={{},$p_{j^*}$,{},$4 p_{j^*}$,$8 p_{j^*}$,$16 p_{j^*}$,$32 p_{j^*}$},
                clip=false,
            ]
            \nextgroupplot[
                ymin=4/3, ymax=17/9,
                height=5/6*5cm,
                ytick={5/3,5/3+1/9},
                yticklabels={
                    $\frac{2 \cdot (\sqrt{m})^{-1}}{1 - (\sqrt{m})^{-1}} + \frac{8 \sqrt{m}}{k_{j^*}}$,
                    {}},
                ylabel = {$\Phi^{j^*}_i$},
                ylabel style={at={(axis description cs:0.11,0.98)}, anchor=south east, rotate=-90},
                every axis plot/.append style={very thick},
                axis x line=none,
                axis y line=left,
            ]
            \addplot [
                color=red,
                ]
                coordinates {
                (0,5/3+1/9)(2,5/3)(2.5 - 1.5*0.035,4/3 + 0.035)
                };
            \nextgroupplot[
                ymin=0, ymax=2/3,
                height=5cm,
                ytick={0,1/9,1/3,5/9},
                yticklabels={
                    0,
                    $\frac{2 \cdot (\sqrt{m})^{-2}}{1 - (\sqrt{m})^{-1}}$,
                    $\frac{2 \cdot (\sqrt{m})^{-1}}{1 - (\sqrt{m})^{-1}}$,
                    $\frac{2}{\sqrt{m}} \cdot \left(1 + \frac{1}{1 - \sqrt{m}^{-1}}\right)$,},
                every axis plot/.append style={very thick},
                axis x line=bottom,
                axis y line=left,
                y axis line style={-},
            ]
            \draw[white,ultra thick,fill=white] (axis cs:-1,2/3 - 0.015) rectangle (axis cs:1,2/3 + 0.035);
            \draw[black,very thick] (axis cs:-.75,2/3 + 0.03) -- ++(axis cs:1.5,0.01); %
            \draw[black,very thick] (axis cs:-.75,2/3 - 0.02) -- ++(axis cs:1.5,0.01); %
            \addplot [
                color=blue,
                ]
                coordinates {
                (0,5/9)(8,1/9)(16,1/27)(32,1/81)(40,1/81 - 1/486)
                };
            \addplot [
                color=red,
                ]
                coordinates {
                (3.5 - 1.5*0.015,2/3 - 0.015)(4,1/3)
                };
            \end{groupplot}
        \end{tikzpicture}
        \caption{Plot of the potential term $\smash{\Phi^{j^*}_i}$ as a function of the $\alpha$-\PropOne-slack. Machines start on the blue graph and switch to the red graph when they receive a job larger than $\alpha \cdot p_{j^*} / k_{j^*}$.}
        \label{fig:chore_online_general_size_plot_phi_j_star_i}
    \end{figure}

    The induction step uses a job-dependent potential $\Phi^{j^*}=\sum_{i\in\M}\Phi^{j^*}_i$, tailored to the current job~$j^*$. As shown in \Cref{fig:chore_online_general_size_plot_phi_j_star_i}, the potential term $\Phi^{j^*}_i$ consists of a main term similar to the potential function in the unit-size analysis, plus a local penalty for machines whose largest job is too large to be shadowed by $j^*$.

    We will show that the stronger invariant holds for $j^*$ as long as $\Phi^{j^*}$ is bounded by its initial value and that $\Phi^{j^*}$ does not increase when any of the previous jobs are allocated.

    For the first statement, let $i$ be a machine on which $j^*$ can neither fit nor shadow. By the fit and shadow conditions, we have $s_{\alpha \PropOne}(i) < 2 p_{j^*}$ and $p_{\max}(i, \sigma) > \alpha p_{j^*}/k_{j^*}$. But then the potential term $\Phi^{j^*}_i$ includes the full penalty term. If this is the case for at least half of the eligible machines, then the total potential from these machines exceeds the initial value of $\Phi^{j^*}$.

    It remains to show that $\Phi^{j^*}$ does not increase when the previous jobs are allocated.

    First, observe that a potential term $\Phi^{j^*}_i$ only increases when the $\alpha$-\PropOne-slack of machine $i$ decreases or when the penalty term is activated while $i$ has rather small $\alpha$-\PropOne-slack $s_{\alpha \PropOne}(i) < 4 p_{j^*}$. Both can only happen when $i$ receives a job.

    Let $i$ be a machine receiving a previous job $j$.

    If the $\alpha$-\PropOne-slack of machine $i$ decreased due to $j$, then $j$ fit on $i$. In particular, $i$ is a machine with maximum $\alpha$-\PropOne-slack among the eligible machines of $j$. Moreover, by the fit condition, the $\alpha$-\PropOne-slack of machine $i$ decreases by at most a factor $2$ due to the load of $j$. Therefore, the increase in $\Phi^{j^*}_i$ is small enough to be well compensated for by the decrease in the potential terms of the other eligible machines.

    For $j$ to activate the penalty term of machine $i$, it needs to be sufficiently large, namely $p_j > \alpha \cdot p_{j^*} / k_{j^*}$. Therefore, $j$ substantially increases the $\alpha$-\PropOne-slacks of the other eligible machines. On the eligible machines $i'$ whose $\alpha$-\PropOne-slack is small beforehand, say $s_{\alpha \PropOne}(i') < 4 p_{j^*}$, this yields a substantial decrease of the potential terms.

    If the potential term of machine $i$ actually increases due to the activation of the penalty term, its $\alpha$-\PropOne-slack is itself smaller than $4 p_{j^*}$. In the fit case, the behavior of our algorithm guarantees that all other eligible machines satisfy the same bound. In the shadow case, the induction hypothesis guarantees the stronger invariant for $j$ which implies that more than half of the eligible machines satisfy the same bound. In both cases, the potential decrease on the respective machines is sufficiently large to pay for the whole penalty term of machine $i$.

    It follows that the potential $\Phi^{j^*}$ does not increase when the previous jobs are allocated. This completes the proof.
\end{proof}

\section{Online Proportionality for Goods}
\label{sec:goods_setting}
We now turn to online goods. In contrast to chores, the optimal approximation factor is logarithmic %
in $m$. We give an $O(\log m)$-approximation and prove that this is best possible, even for unit values.

\begin{restatable}{theorem}{ThmLogmGoods}
\label{thm:log-m-upper-bound-online-goods}
There is a deterministic online algorithm for restricted-assignment goods allocation that achieves $O(\log m)$-\PropOne{} for arbitrary values.
\end{restatable}

The algorithm is built in two steps. We first analyze unit-value goods. The key rule is an analogue of Greedy $\alpha$-slack assignment for unit-value chores from \Cref{sec:chores-unit}: instead of assigning to a machine with maximum $\alpha$-slack, we assign to an agent with maximum $\alpha$-deficit. Despite this similarity, the achievable guarantee is exponentially better. We then extend the result to arbitrary values by a natural bucketing scheme.

\medskip
\noindent\textbf{Greedy \boldmath$\alpha$-deficit assignment.}
The case $m = 1$ is trivial.
For $m \geq 2$, let $\alpha = \log_2(m)$.
When a good $j$ arrives, assign it to an eligible agent $i\in\M(j)$ with
maximum $\alpha$-deficit $d_\alpha(i) \coloneqq \frac1\alpha\PS_i-L_i$, where $\PS_i$ includes the arriving good $j$ but not
future goods.

Weighing the load higher than the proportional share is important: using the natural $1$-deficit $\PS_i-L_i$ can be as bad as $\Theta(m)$; see \Cref{lem:goods_online_specific_lower_bound_slack_greedy} in \Cref{app:deficit-greedy-fails}.

Our analysis uses the following technical lemma whose proof we defer to \cref{app:goods-technical-lemma}.
\begin{restatable}{lemma}{TechnicalLemmaGoodsOnline}
    \label{lem:goods-online-show-potential-difference-factor-negative}
    Let $k \geq 2$ and $m \geq 3$ be integers such that $m \geq k$. Then
    \[
        (k - 1) \cdot 2^{\frac{1}{k}} + \frac{2^{\frac{1}{k}}}{m} - k < 0.
    \]
\end{restatable}
We now analyze greedy $\alpha$-deficit assignment for equal-value goods:

\begin{theorem}
\label{thm:goods-unit-upper}
For restricted-assignment online allocation of equal-value goods to $m \geq 2$ agents, Greedy $\alpha$-deficit assignment achieves
$\alpha$-\PropOne{} for $\alpha = \log_2(m)$.
\end{theorem}
\begin{proof}
    For $m = 2$ agents, greedy $1$-deficit allocation clearly maintains \PropOne.

    Let $m \geq 3$.
    By scaling, assume that all goods have value $1$. We use a potential function that exponentially penalizes agents for which the $\alpha$-deficit $d_\alpha(i)$ is large. We define $\Phi \coloneqq \sum_{i \in \M} \Phi_i$ with
    \begin{align*}
        \Phi_i &\coloneqq m^{d_\alpha(i)}. %
    \end{align*}

    We will show that $\Phi$ does not increase when a new good arrives and is allocated to the eligible agent with largest $\alpha$-deficit.
    From this, it follows that every configuration that can occur during the algorithm's execution satisfies $\Phi \leq m$. In particular, it holds for every agent $i$ that $\Phi_i \leq m$ which we can rewrite as
    \[
        \frac{1}{\alpha}\PS_i - L_i \leq 1.
    \]
    As every good has value $1$, this is exactly the desired $\alpha$-\PropOne{} guarantee.

    It remains to show that $\Phi$ does not increase when a new good arrives.
    Let $i^*$ be the agent to which a good $j$ is allocated.
    Then $j$ increases the proportional shares of all agents in $\M(j)$ by $\frac{1}{k_j}$ and the load of agent $i^*$ by $1$.
    The other loads and proportional shares stay unchanged.
    Thus, $\Phi$ changes as follows:
    \[
        \Phi^{\text{new}} - \Phi^{\text{old}} = \sum_{i \in \M(j) \setminus \{i^*\}} \left(m^{d_\alpha(i) + \frac{1}{\alpha \cdot k_j}} - m^{d_\alpha(i)}\right) + m^{d_\alpha(i^*) + \frac{1}{\alpha \cdot k_j} - 1} - m^{d_\alpha(i^*)}
    \]
    By the greedy criterion, $i^*$ is an agent for which $d_\alpha(i)$ is maximum.
    Therefore, we can bound the change in $\Phi$ as follows:
    \begin{align*}
        \Phi^{\text{new}} - \Phi^{\text{old}} &= \sum_{i \in \M(j) \setminus \{i^*\}} \left(m^{d_\alpha(i) + \frac{1}{\alpha \cdot k_j}} - m^{d_\alpha(i)}\right) + m^{d_\alpha(i^*) + \frac{1}{\alpha \cdot k_j} - 1} - m^{d_\alpha(i^*)}\\
        &\leq \sum_{i \in \M(j) \setminus \{i^*\}} m^{d_\alpha(i^*)} \cdot \left(m^{\frac{1}{\alpha \cdot k_j}} - 1\right) - m^{d_\alpha(i^*)} \cdot \left(1 - m^{\frac{1}{\alpha \cdot k_j} - 1}\right)\\
        &= m^{d_\alpha(i^*)} \cdot \left((k_j - 1) \cdot m^{\frac{1}{\alpha \cdot k_j}} + m^{\frac{1}{\alpha \cdot k_j} - 1} - k_j\right)\\
        &= m^{d_\alpha(i^*)} \cdot \left((k_j - 1) \cdot 2^{\frac{1}{k_j}} + \frac{2^{\frac{1}{k_j}}}{m} - k_j\right).
    \end{align*}
    By \cref{lem:goods-online-show-potential-difference-factor-negative}, the second factor is negative since we have $m \geq 3$.
    Hence, $\Phi$ does not increase when a new good arrives and is allocated to the eligible agent with largest $\alpha$-deficit. This completes the proof.
\end{proof}

We extend the equal-value result to arbitrary values by bucketing goods according to rounded powers of two; the details are deferred to \Cref{app:ub-online-goods}.

\paragraph{Bucketing algorithm.} (for arbitrary values): Upon the arrival of a good $j$ with utility~$p_j$, we round its value up to the nearest power of $2$, denoted by $p'_j = 2^{\lceil \log_2 p_j \rceil}$. We maintain a separate instance for each subset of goods with equal rounded utilities and apply the greedy $\alpha$-deficit algorithm for equal-value goods to each virtual instance.

\medskip
We complement the upper bound with a matching lower bound.

\begin{theorem} \label{thm:good_setting_online_lb}
    No deterministic online algorithm %
    can guarantee to find a restricted-assignment goods allocation that is $\alpha$-\PropOne{} for any $\alpha<\ln(m) - 1$,
    even if all items have utility $1$.
\end{theorem}
\begin{proof}
    Fix a deterministic online algorithm.
    Consider the following adversarial instance with unit-value goods.
    The first good is compatible with all agents.
    After seeing which agent receives it, the adversary releases a second good compatible with all other agents. After seeing which agent receives the second good, the adversary releases a third good which is compatible with all agents except the agents that received the first two goods and so on. The adversary proceeds in this way until one agent remains that has not received a good.
    This remaining agent does not receive any good and their largest compatible good has value $1$. However, their proportional share is equal to
    \begin{align*}
        \frac{1}{m} + \frac{1}{m - 1} + \dots + \frac{1}{2} &= H_m - 1 \geq \ln(m) - 1. \qedhere
    \end{align*}
\end{proof}

\section{Envy-based Fairness under Eligibility Constraints}
\label{sec:envy}
We now turn to envy-based fairness. Under eligibility constraints, envy comparisons are not canonical: a job assigned to one machine may be irrelevant to another machine that could never process it. We therefore consider two natural eligibility-aware notions. The strong notion compares two machines only through jobs eligible for both; the weak notion ignores, on the envying side, jobs that the envied machine could not process.

The following definitions are the ones used throughout the section. For machines $i,i'\in\M$, let $C(i,i') \coloneqq \{j \in [n]:i,i'\in\M(j)\}$ be the set of jobs for which both machines are eligible. In the strong notion, $i$ compares only the parts of the two bundles consisting of jobs in $C(i,i')$. In the weak notion for chores, $i$ compares all jobs assigned to $i'$ with only those jobs of $i$ that $i'$ could also have received.

\begin{definition}[Strong envy]
For two machines $i,i' \in \M$, let
$
C(i,i') \coloneqq \{j \in [n] : i,i' \in \M(j)\}
$
denote the set of jobs for which both machines are eligible. Given an allocation
$\sigma$, define
$
L_{i,i'}(\sigma)
\coloneqq \sum_{\substack{\sigma(j)=i\\ j \in C(i,i')}} p_j .
$
Machine $i$ \emph{strongly envies} machine $i'$ if
$
L_{i,i'}(\sigma) >
L_{i',i}(\sigma).
$
\end{definition}

An allocation is \emph{strongly envy-free} if no machine strongly envies another machine. It satisfies \emph{strong EF1} if, for every pair $i,i'\in\M$, either $i$ does not strongly envy $i'$, or there is a job $j$ assigned to $i$ with $j\in C(i,i')$ such that removing $j$ eliminates the strong envy. It satisfies \emph{strong EFX} if the same holds for every job $j$ assigned to $i$ with $j\in C(i,i')$.

\begin{definition}[Weak envy]
Let $i,i' \in \M$.
We say that machine~$i$ \emph{weakly envies} machine $i'$ if
$
L_{i,i'}(\sigma) > L_{i'}(\sigma).
$
\end{definition}

An allocation is \emph{weakly envy-free} if no machine weakly envies another
machine. It satisfies \emph{weak EF1} if, for every pair $i,i' \in \M$, either $i$ does not
weakly envy $i'$, or there is a job $j$ assigned to $i$ with $i'\in\M(j)$ such
that removing $j$ eliminates the weak envy. It satisfies \emph{weak EFX} if the same
holds for every job $j$ assigned to $i$ with $i'\in\M(j)$.

The envy definitions for the goods setting are analogous.
\begin{definition}[Envy for goods]
    Given an allocation $\sigma$ of goods, agent $i$ strongly envies agent $i'$ if
    $
    L_{i,i'}(\sigma) < L_{i',i}(\sigma).
    $
    Agent $i$ weakly envies agent $i'$ if
    $
    L_{i}(\sigma) < L_{i',i}(\sigma).
    $
\end{definition}

Again, an allocation is \emph{strongly (weakly) envy-free} if no agent strongly (weakly) envies another agent. It satisfies \emph{strong (weak) EF1} if, for every pair $i, i' \in \M$, either $i$ does not strongly (weakly) envy $i'$, or there is a good $j$ assigned to $i'$ with $j \in C(i, i')$ such that removing $j$ eliminates the strong (weak) envy.

\paragraph{Remark.}
In the graphical and restricted-additive models discussed in \Cref{sec:related-work}, envy is measured with the agent's own valuation, which is zero on items the agent is not eligible for. For goods under identical valuations, envy and EF1 in that reading coincide with our weak notions, since $v_i(\sigma^{-1}(i')) = L_{i',i}(\sigma)$; in particular, the offline existence of weak EF1 also follows from~\cite{DEGK25orientations}. EFX does not coincide: there, envy must vanish after removing \emph{any} good from the envied bundle, including goods the envying agent values at zero~\cite[Lemma~2.11]{KSS24tale}, so that an agent may not envy at all a bundle containing a good she is not eligible for. Our weak EFX removes only goods the envying machine is eligible for and is therefore weaker; our results do not bear on the stronger variant. For chores, the valuation-based reading yields yet another notion, incomparable in spirit to ours since it charges a machine for jobs the envied machine could never have taken; EF1 orientations may fail to exist for chores in that reading~\cite{SWL26orientations}, whereas weak EF1 always exists in ours.

For envy-based notions, identical valuations make goods and chores equivalent, even under eligibility constraints. This contrasts with the proportionality-based notions studied above.

\begin{restatable}[Equivalence of goods and chores for envy-based notions]{proposition}{PropEquivalenceGoodsChores}
\label{prop:envy-goods-chores-equivalence}
Consider a restricted-assignment chore instance with jobs $J$ and processing times $(p_j)_{j\in J}$. Replace each chore $j$ by a good of value $p_j$, preserving the eligibility sets and the allocation. Then an allocation satisfies strong, respectively weak, EF, EF1, or EFX in the chore instance if and only if the corresponding allocation satisfies the same strong, respectively weak, notion in the goods instance.
\end{restatable}

This equivalence contrasts sharply with proportionality-based notions, where goods and chores lead to different guarantees; see
\Cref{sec:chore_online,sec:goods_setting}. In particular, simple
examples show that a \PropOne{} allocation for goods need not be \PropOne{} for chores, and vice versa; see \Cref{obs:propone-goods-neq-chores}.

The proof is by reversing the direction of envy and is deferred to \Cref{app:goods=chores}. In the remainder of this section, we state the results for chores only; by \Cref{prop:envy-goods-chores-equivalence}, they transfer directly to goods.

Together, the results in this section show that strong guarantees with respect to envy-based fairness notion are largely impossible under eligibility constraints, whereas weak EF1 remains achievable online and weak EFX is achievable offline but not online. For convenience, \Cref{fig:envy-notions} summarizes the envy-based results from this section. %

\providecolor{badBorder}{HTML}{C0392B}
\providecolor{badFill}{HTML}{FDF0EF}
\providecolor{midBorder}{HTML}{B9770E}
\providecolor{midFill}{HTML}{FEF8EC}
\providecolor{goodBorder}{HTML}{1E8449}
\providecolor{goodFill}{HTML}{EDF9F1}
\providecolor{arrowGray}{HTML}{5D6D7E}
\providecolor{inkDark}{HTML}{2C3E50}

\begin{figure}[t]
\centering
\begin{tikzpicture}[
    node distance=0.45cm and 0.55cm,
    notion/.style={
        draw, rounded corners=3pt, align=center,
        minimum width=2.75cm, minimum height=0.85cm,
        line width=0.8pt, inner ysep=3.5pt, inner xsep=5pt,
        drop shadow={shadow xshift=0.6pt, shadow yshift=-0.9pt,
                     opacity=0.15, fill=black},
        fill=white
    },
    nonexist/.style={notion, draw=badBorder,  fill=badFill},
    offline/.style ={notion, draw=midBorder,  fill=midFill},
    both/.style    ={notion, draw=goodBorder, fill=goodFill},
    implies/.style={
        -{Stealth[length=2.4mm, width=1.8mm]},
        line width=0.9pt, draw=arrowGray,
        shorten >=2pt, shorten <=2pt
    }
]

\node[nonexist] (sefx) {%
    {\footnotesize\bfseries\color{inkDark} Strong EFX}\\[-5pt]
    {\tiny\color{badBorder} $\times$\, does not exist}};

\node[nonexist, below left=of sefx] (sef1) {%
    {\footnotesize\bfseries\color{inkDark} Strong EF1}\\[-5pt]
    {\tiny\color{badBorder} $\times$\, does not exist}};

\node[offline, below right=of sefx] (wefx) {%
    {\footnotesize\bfseries\color{inkDark} Weak EFX}\\[-5pt]
    {\tiny {\color{goodBorder}\checkmark\, offline} \;\,
           {\color{badBorder}$\times$\, online}}};

\node[both, below right=of sef1] (wef1) {%
    {\footnotesize\bfseries\color{inkDark} Weak EF1}\\[-5pt]
    {\tiny\color{goodBorder} \checkmark\, offline \;\, \checkmark\, online}};

\draw[implies] (sefx) -- (sef1);
\draw[implies] (sefx) -- (wefx);
\draw[implies] (sef1) -- (wef1);
\draw[implies] (wefx) -- (wef1);

\end{tikzpicture}
\caption{The four eligibility-aware envy notions, ordered by logical strength ($A \rightarrow B$ means that every allocation satisfying $A$ also satisfies $B$), together with our offline and online existence results.} %
\label{fig:envy-notions}
\end{figure}
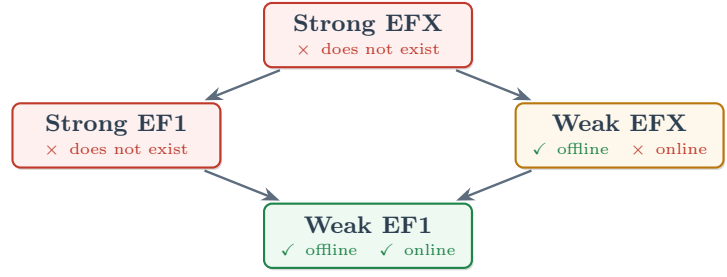

\subsection{Strong Envy}

The strong notion is arguably the closest analogue of envy under eligibility constraints: machine $i$ compares itself to machine $i'$ only through jobs that both machines could have processed. However, this notion is too demanding: strong EF1 may fail to exist even offline, and does not admit any approximation guarantee in the online setting. Proofs %
are deferred to \Cref{app:strong-ef1}.

\begin{restatable}{theorem}{ThmNoEFOneChore}
\label{thm:noEF1-chore}
In restricted-assignment chore allocation with general processing times,
strong EF1 is not always achievable, even when $m=7$.
\end{restatable}
\begin{proof}[Proof sketch]
We construct an instance using the complement of the Fano plane:  %
take seven unit-size jobs and seven machines, and choose the eligibility sets to be the seven $4$-element complements of the lines of the Fano plane. Equivalently, these eligibility sets form the unique symmetric $2$-$(7,4,2)$ design: every job has four eligible machines, every machine is eligible for four jobs, any two jobs have exactly two common eligible machines, and any two machines are eligible for exactly two common jobs. We then add two jobs of size $p\in(1/2,1)$ for which all machines are eligible.

The incidence structure forces any strong-EF1 allocation into a very rigid
form. If some machine receives two unit-size jobs, then the unique common
compatibilities of these jobs force both small jobs to be placed in a way that creates strong envy elsewhere. Hence any strong-EF1 allocation would have to assign exactly one unit-size job to each machine. In that case, the directed strong-envy graph induced by the unit-size jobs is almost fully covered by cycles. Since there are only two small jobs, at least one machine on such a cycle receives no small job while its predecessor on the cycle does, causing the predecessor to strongly envy it even after removing one job. Thus no allocation satisfies strong EF1.
\end{proof}

In the online setting, %
even finite %
approximations of strong EF1 are impossible to achieve
for general processing times.
For $\alpha>0$, an allocation satisfies strong $\alpha$-EF1 if, for every pair $i,i'\in\M$, either $i$ does not strongly envy $i'$, or there is a job $j\in C(i,i')$ assigned to $i$ such that $L_{i,i'}-p_j\le \alpha L_{i',i}$, where $L_{i,i'}$ denotes the load assigned to $i$ among jobs compatible with both $i$ and $i'$.

\begin{restatable}{theorem}{ThmNoAlphaEFOne}
For every $\alpha>0$, no deterministic online algorithm can guarantee strong $\alpha$-EF1 in restricted-assignment chore allocation with general processing times, even when $m=3$.
\end{restatable}

\subsection{Weak Envy}
The weak notion is more permissive: jobs incompatible with the envied machine are ignored on the envying side, while the full load of the envied machine is still counted. This guarantees the existence of EFX allocations offline and makes EF1 (but not EFX) achievable online.
The online results largely carry over from the unrestricted temporal setting for which they were shown by Elkind et al.~\cite{ELL+25temporal}; we show that offline EFX can be achieved by sorting the jobs in order of non-increasing processing times and then running the online algorithm.
For completeness, we give full proofs for all results in \Cref{app:weak-envy}.

\begin{restatable}{theorem}{ThmWeakEFOneGreedy}
There is a deterministic online algorithm that maintains weak EF1 for restricted-assignment chore allocation for arbitrary processing times.
An allocation satisfying weak EFX exists and can be computed offline in polynomial time.
\end{restatable}

\begin{restatable}%
{theorem}{ThmNoWeakEFX}
No deterministic online algorithm can guarantee weak EFX in
restricted-assignment chore allocation, even when $m=2$.
\end{restatable}

\section{Conclusion}
\label{sec:conclusion}

We initiate the study of online proportional and envy-based fairness under restricted assignment constraints. Eligibility creates a sharp chore--goods separation for proportionality, while the choice of eligibility-aware notion becomes decisive for envy-based fairness. Several questions remain open.

The offline results identify \PropOne{} as the strongest exact proportionality notion that is always achievable under restricted assignment~(\Cref{sec:offline}; see also~\cite{SWL26orientations}): \Prop{} and \PropX{} may fail, while \PropOne{} always exists. This leaves open the best possible approximation factors for the stronger notions. In particular, what are the optimal guarantees for $\alpha$-\Prop{} and $\alpha$-\PropX{} in the offline and online settings?

For envy-based fairness, our results show that strong and weak notions behave very differently: strong EF1 may fail to exist offline, whereas weak EF1 is achievable online and weak EFX separates offline from online. A natural question is whether strong EF1 admits meaningful approximation guarantees, or whether alternative eligibility-aware envy notions interpolate between the strong and weak definitions. More broadly, it would be an interesting research direction to investigate meaningful fairness notions for other constraint settings.

Our online results are for adversarial arrivals and irrevocable assignments. It would be interesting to understand whether stronger guarantees become possible under additional information or milder online models, such as randomized algorithms, random-order arrivals, predictions of maximum item values~\cite{choo2025approximateproportionalityonlinefair,NeohT2026closing,DuS2026online}, known final proportional shares, or limited recourse~\cite{AmanatidisGMP2026buffers}; for known final proportional shares, we show that this may not help in the online chore setting~(\Cref{thm:chore-online-LB}). Such models may help separate the difficulty caused by online uncertainty from the difficulty caused by eligibility constraints.

\paragraph{Acknowledgments.} We thank Sander Borst, Sarah Morell, and Karol Węgrzycki for helpful discussions on this topic. This work was supported by the Deutsche Forschungsgemeinschaft (DFG, German Research Foundation) under Germany's Excellence Strategy – EXC 3036 – project no.\ 533607631.

\paragraph{Use of AI tools.}
The authors used an AI assistant for language editing and for checking the related-work discussion. All results, proofs, and technical content are the authors' own, and all suggested edits were reviewed and approved by the authors.

\newpage
\appendix
\crefalias{section}{appendix}
\crefalias{subsection}{appendix}
\crefalias{subsubsection}{appendix}

\section*{Appendices}
\section{Omitted Proofs from \Cref{sec:preliminaries} (Offline Goods and Chores)}
\label{app:offline}

We collect the omitted proofs for the offline proportionality results from \Cref{sec:offline}. Throughout this appendix, the full instance is known in advance and a single allocation must be computed.

Although the goods and chore settings share their notation, their fairness notions do not coincide: an allocation may be fair in one reading and unfair in the other, in either direction.

\ObsPropOneGoodsNeqChores*
\begin{proof}
    For the first statement, consider $3$ agents and $4$ unit-size chores for which all agents are eligible.
    Allocating two chores each to two agents satisfies \PropOne, but this allocation would not satisfy \PropOne{} if the chores were goods.

    For the second statement, consider $3$ agents and $2$ unit-size goods for which all agents are eligible.
    Assigning both goods to the same agent satisfies \PropOne, but this allocation would not satisfy \PropOne{} if the goods were chores.
\end{proof}

Next we show that, even in this fully informed setting, the stronger proportionality notions are not always achievable: there are instances with no \PropX{} allocation, and hence with no exact proportional allocation. We then show that \PropOne{} remains achievable and can be computed in polynomial time.

\ObsOfflineNoPropX*
\begin{proof}
For the chore setting, consider two machines $i_1,i_2$ and three jobs $j_1,j_2,j_3$ with $p_{j_1}=1$ and $p_{j_2}=p_{j_3}=\varepsilon$, where $\varepsilon\in(0,1/2)$, and let $\M(j_1)=\{i_1,i_2\}$, $\M(j_2)=\{i_1\}$, and $\M(j_3)=\{i_2\}$.

Then both machines have proportional share $\PS_{i_1}=\PS_{i_2}=1/2+\varepsilon$.
The jobs $j_2$ and $j_3$ are forced to machines $i_1$ and $i_2$, respectively, and $j_1$ must be assigned to one of them. Suppose, without loss of generality, that $j_1$ is assigned to $i_1$. Then $L_{i_1}=1+\varepsilon$, and removing the small job
$j_2$ leaves load
\[
L_{i_1}-p_{j_2}=1 \;>\; 1/2+\varepsilon=\PS_{i_1},
\]
which fails because $\varepsilon<1/2$. Thus $\PropX{}$ is violated; the case in which $j_1$ is assigned to $i_2$ is symmetric. Hence no allocation satisfies \PropX{}. Since every \Prop{} allocation is also \PropX{}, the same instance rules out \Prop{} as well.

For the goods setting, consider three machines $i_1,i_2,i_3$ and two jobs $j_1,j_2$ with $p_{j_1}=1$ and $p_{j_2}=\varepsilon$, where $\varepsilon\in(0,1/3)$, and let $\M(j_1)=\M(j_2)=\{i_1,i_2,i_3\}$.
Then every machine has proportional share $\PS_{i}=1/3+\varepsilon/3$.
Since there are only two jobs, in every allocation some machine $i$ receives no job, so $L_i=0$. Both jobs are eligible for $i$ and allocated elsewhere, so $\PropX{}$ requires in particular
\[
L_i+p_{j_2}=\varepsilon \;\ge\; 1/3+\varepsilon/3=\PS_i,
\]
which fails. Thus no allocation satisfies $\PropX{}$. Since every \Prop{} allocation is also \PropX{}, the same instance rules out \Prop{} as well.
\end{proof}
\ThmOfflinePropOne*
\begin{proof}
We use the rounding argument of Lenstra, Shmoys, and
Tardos~\cite{lenstra1990approximation} for chores and the modification by Bez{\'a}kov{\'a} and Dani~\cite{bezakova2005allocating} for goods.
Machines with $\PS_i=0$ are not eligible for any job and can be ignored. For every job $j$ and every eligible machine $i$, define the scaled processing time $q_{ij} \coloneqq p_j/\PS_i$.

In the chore setting, consider the assignment LP with machine bounds normalized to one:
\begin{align*}
\sum_{i\in\M(j)} x_{ij} &= 1 && \text{for all } j,\\
\sum_{j:i\in\M(j)} q_{ij}x_{ij} &\le 1 && \text{for all } i,\\
x_{ij} &\ge 0 && \text{for all } i\in\M(j).
\end{align*}
This LP is feasible, since the fractional allocation $x_{ij}=1/|\M(j)|$ gives,
for every machine $i$,
\[
\sum_{j:i\in\M(j)} q_{ij}x_{ij}
= \frac{1}{\PS_i}\sum_{j:i\in\M(j)} \frac{p_j}{|\M(j)|}
=1.
\]

The rounding theorem of Lenstra, Shmoys, and Tardos rounds a basic solution of this LP in polynomial time to an integral assignment $\sigma$ such that, for every machine $i$,
\[
\sum_{j:\sigma(j)=i} q_{ij}\le 1+\max_{j:\sigma(j)=i} q_{ij}.
\]
Multiplying by $\PS_i$ yields
\[
L_i(\sigma)=\sum_{j:\sigma(j)=i}p_j
\le \PS_i + p_{\max}(i, \sigma).
\]
Equivalently, $L_i(\sigma) - p_{\max}(i, \sigma) \le \PS_i$ for every machine
$i$, so $\sigma$ satisfies \PropOne{}. The rounding algorithm runs in polynomial
time.

In the goods setting, the proportionality constraints in the assignment LP change direction and become
\begin{align*}
    \sum_{j:i\in\M(j)} q_{ij}x_{ij} &\ge 1 && \text{for all } i.
\end{align*}
As in the chore setting, the fractional allocation $x_{ij}=1/|\M(j)|$ is a feasible solution.
The rounding theorem by Lenstra, Shmoys and Tardos with the modification of Bez{\'a}kov{\'a} and Dani rounds a basic solution in polynomial time to an integral assignment $\sigma$ such that every agent $i$ receives all goods that are partially allocated to them in the basic solution except for at most one such good $j$.
It follows that
\[
    \sum_{j':\sigma(j')=i} q_{ij'}\ge 1-q_{ij} \ge 1-\max_{j':\sigma(j')\neq i} q_{ij'}
\]
for every agent $i$ which implies that $\sigma$ satisfies \PropOne{}. The algorithm runs in polynomial time.
\end{proof}

\section{Omitted Proofs from \cref{sec:chore_online} (Online Chore Setting)}
\label{app:chore-online}

\subsection{Simple Greedy strategy fails}
\label{app:greedy-fails}

While restricted-assignment techniques transfer surprisingly well in the offline case, the online setting exhibits a sharp separation from classical restricted assignment.
This separation is already visible for natural greedy rules.
In online makespan minimization for restricted assignment, assigning each job to a least-loaded eligible machine gives an $O(\log m)$-competitive algorithm.
A natural fairness analogue is to assign each arriving job to an eligible machine that would have maximum slack towards violating \PropOne{} after allocating the new job. This \PropOne-slack is given by
\[
    \PS_i - L_i + p_{\max}(i, \sigma).
\]
The following lemma shows that this rule fails dramatically:
even for unit processing times, it cannot achieve an $o(m)$-approximation.

\begin{lemma}\label{lem:greedy-fails}
    Let $\alpha \leq \frac{2}{15}m$. The algorithm that assigns each arriving %
    job to an eligible machine with maximum hypothetical $\PropOne{}$-slack, computed after hypothetically allocating the arriving job to the respective machine, does not achieve $\alpha$-\PropOne{}, even for unit-size jobs. %
\end{lemma}
\begin{proof}
    Let $m$ be any number of machines that is divisible by $5$.
    Consider the following instance with unit processing times on $m$ machines. First, let $\frac{2}{5} m$ jobs arrive for which only the first $\frac{4}{5} m$ machines are eligible.
    By the greedy criterion, each of these jobs will be assigned to a different machine; assume w.l.o.g.\ that they are assigned to the first $\frac{2}{5}m$ machines. At this stage, each of the first $\frac{4}{5}m$ machines has proportional share $\frac{1}{2}$.

    Next, for each machine $i^*$ among the last $\frac{1}{5} m$ machines, release one job whose eligible machines are the first $\frac{2}{5} m$ machines and $i^*$. These jobs increase the proportional shares of the first $\frac{2}{5} m$ machines by less than $\frac{1}{2}$ in total. Allocating such job to a machine $i$ among the first $\frac{2}{5}m$ machines would give
    \begin{align*}
        \PS_i - L_i + p_{\max}(i, \sigma) &< 1 - 2 + 1 %
        = 0.
    \end{align*}
    In contrast, assigning it to the corresponding machine $i^*$ gives
    \begin{align*}
        \PS_{i^*} - L_{i^*} + p_{\max}(i^*, \sigma) &= \frac{1}{\frac{2}{5} m + 1} - 1 + 1 %
        > 0.
    \end{align*}
    By the greedy criterion, each of these jobs is thus assigned to the respective machine $i^*$.

    Finally, release one job for which only the last $\frac{1}{5}m$ machines are eligible.
    The machine receiving this job has two jobs in total, and after removing its largest job its remaining load is $1$. Its proportional share is
    \begin{align*}
        \frac{1}{\frac{2}{5} m + 1} + \frac{1}{\frac{1}{5} m} &< \frac{15}{2m}.
    \end{align*}
    Since $\alpha \leq \frac{2}{15}m$, we have $\alpha \cdot \PS_i < 1$, and therefore the allocation violates
    $\alpha$-$\PropOne$.
\end{proof}

\subsection{Upper bound online chores}
\label{app:general-upper-bound-online-chores}

\ThmChoreOnlineGeneral*

Before proving \cref{thm:general_upper_bound}, we describe again our algorithm, in particular how we handle jobs with at most $\alpha$ eligible machines.

\paragraph{The Algorithm.} We handle every new job $j^*$ as follows:

If $k_{j^*}\le \alpha$, we assign $j^*$ to an arbitrary eligible machine. Such jobs are ignored by the algorithm in all future computations; equivalently, the algorithm maintains a virtual instance consisting only of jobs with more than $\alpha$ eligible machines. 

If $k_{j^*}>\alpha$, we first compute the hypothetical $\alpha$-\PropOne{}-slack of every eligible machine after adding $\alpha p_{j^*}/k_{j^*}$ to its scaled proportional share. Let $i$ be an eligible machine with maximum updated $\alpha$-\PropOne{}-slack. We distinguish two cases.

\begin{enumerate}
    \item \textbf{Case 1 (Fit):} If the largest updated slack satisfies $s_{\alpha\PropOne}(i) \geq 2 \cdot p_{j^*}$, we assign~$j^*$ to machine $i$.
    We say that $j^*$ \emph{fits} on $i$.

    \item \textbf{Case 2 (Shadow):} Otherwise, consider the eligible machines with $p_{\max}(i, \sigma) \leq \alpha p_{j^*}/k_{j^*}$.
    Among these machines, assign $j^*$ to one with maximum updated $\alpha$-\PropOne{}-slack.
    We say that~$j^*$ \emph{shadows} the previous largest job on that machine.
\end{enumerate}

A job $j^*$ with $k_{j^*} \le \alpha$ is harmless: on every eligible machine, its contribution to the scaled proportional share is $\alpha p_{j^*}/k_{j^*}\ge p_{j^*}$, so the receiving machine's $\alpha$-\PropOne{}-slack can never be smaller without it.
Such jobs could however interfere with our analysis, for example by activating the penalty term of the receiving machine.
Just like our algorithm, we therefore ignore the impact of these jobs on the $\alpha$-\PropOne-slacks and the values $p_{\max}(i)$ of all machines until we reintroduce these jobs to formally show their harmlessness.
In particular, we abuse notation and denote by $s_{\alpha\PropOne}(i)$ and $p_{\max}(i)$ the values these terms would have if the jobs with at most $\alpha$ eligible machines did not exist.

Most of the proof of \cref{thm:general_upper_bound} is done in the following lemma:

\begin{lemma}
\label{lem:general-shadow-exists}
Let $\alpha=8\sqrt m$, and consider an arriving job $j^*$ with $k_{j^*}>\alpha$. If no eligible machine has updated $\alpha$-\PropOne{}-slack $s_{\alpha\PropOne}(i) \geq 2p_{j^*}$, then more than half of the machines $i\in\M(j^*)$ satisfy $p_{\max}(i)\le \alpha p_{j^*}/k_{j^*}$.
\end{lemma}
\begin{proof}
    As jobs with at most $\alpha$ eligible machines do not influence the behavior of the algorithm, we may assume that every job has more than $\alpha$ eligible machines.
    In particular, since every job has at most $m$ eligible machines, we may assume that $\alpha < k_{j^*} \leq m$ which implies $m > 64$.

    We inductively prove the lemma by showing that the statement holds for $j^*$ if it held for all previous jobs.

    For the induction step, we use a potential function $\smash{\Phi^{j^*} = \sum_{i \in \M} \Phi^{j^*}_i}$ tailored to the job $j^*$.
    The potential term $\Phi^{j^*}_i$ of each machine $i$ consists of a \textquote{main term} and a \textquote{penalty term}:
    \begin{align*}
        \rho(t) &\coloneqq \begin{cases}
            1 & t \leq 8\\
            1 + \left\lceil \log_{2}\left(\frac{t}{8}\right) \right\rceil & t > 8
        \end{cases}\\
        \Phi^{j^*}_i &\coloneqq \int_{s_{\alpha\PropOne}(i) / p_{j^*}}^\infty (2 \sqrt{m})^{- \rho(t)} + \left(4 \cdot \frac{\sqrt{m}}{k_{j^*}} - (2 \sqrt{m})^{-1}\right) \cdot \mathds{1}_{t \in (2, 4)} \cdot \mathds{1}_{p_{\max}(i) > \frac{\alpha \cdot p_{j^*}}{k_{j^*}}} dt.
    \end{align*}
    For convenience, we repeat \Cref{fig:chore_online_general_size_plot_phi_j_star_i_appendix} which shows a plot of the term $\Phi^{j^*}_i$ as a function of the $\alpha$-$\PropOne$-slack $s_{\alpha \PropOne}(i)$.

    \begin{figure}[!t]
        \centering
        \begin{tikzpicture}[scale=0.8]
            \colorlet{blue}{blue!80!black}
            \colorlet{red}{red!85!black}
            \begin{groupplot}[
                group style={group size=1 by 2, vertical sep=0pt},
                width=13cm,
                axis line style={-{Latex}, thick},
                tick style={black, thick},
                major tick style={black, thick},
                minor tick style={black, thick},
                tick label style={black},
                xlabel = {$s_{\alpha \PropOne}(i)$},
                xlabel style={at={(axis description cs:1,.05)}, anchor=west},
                xmin=0, xmax=40,
                xtick={0,1,2,4,8,16,32},
                xticklabels={{},$p_{j^*}$,{},$4 p_{j^*}$,$8 p_{j^*}$,$16 p_{j^*}$,$32 p_{j^*}$},
                clip=false,
            ]
            \nextgroupplot[
                ymin=4/3, ymax=17/9,
                height=5/6*5cm,
                ytick={5/3,5/3+1/9},
                yticklabels={
                    $\frac{2 \cdot (\sqrt{m})^{-1}}{1 - (\sqrt{m})^{-1}} + \frac{8 \sqrt{m}}{k_{j^*}}$,
                    {}},
                ylabel = {$\Phi^{j^*}_i$},
                ylabel style={at={(axis description cs:0.11,0.98)}, anchor=south east, rotate=-90},
                every axis plot/.append style={very thick},
                axis x line=none,
                axis y line=left,
            ]
            \addplot [
                color=red,
                ]
                coordinates {
                (0,5/3+1/9)(2,5/3)(2.5 - 1.5*0.035,4/3 + 0.035)
                };
                \addlegendentry{$p_{\text{max}}(i) > \alpha \cdot p_{j^*} / k_{j^*}$}
            \addplot [color=blue,draw=none] %
                coordinates {(0,5/3+1/9)(2,5/3)(2.5 - 1.5*0.015,4/3 + 0.015)};
                \addlegendentry{$p_{\text{max}}(i) \leq \alpha \cdot p_{j^*} / k_{j^*}$}
            \nextgroupplot[
                ymin=0, ymax=2/3,
                height=5cm,
                ytick={0,1/9,1/3,5/9},
                yticklabels={
                    0,
                    $\frac{2 \cdot (\sqrt{m})^{-2}}{1 - (\sqrt{m})^{-1}}$,
                    $\frac{2 \cdot (\sqrt{m})^{-1}}{1 - (\sqrt{m})^{-1}}$,
                    $\frac{2}{\sqrt{m}} \cdot \left(1 + \frac{1}{1 - \sqrt{m}^{-1}}\right)$,},
                every axis plot/.append style={very thick},
                axis x line=bottom,
                axis y line=left,
                y axis line style={-},
            ]
            \draw[white,ultra thick,fill=white] (axis cs:-1,2/3 - 0.015) rectangle (axis cs:1,2/3 + 0.035);
            \draw[black,very thick] (axis cs:-.75,2/3 + 0.03) -- ++(axis cs:1.5,0.01); %
            \draw[black,very thick] (axis cs:-.75,2/3 - 0.02) -- ++(axis cs:1.5,0.01); %
            \addplot [
                color=blue,
                ]
                coordinates {
                (0,5/9)(8,1/9)(16,1/27)(32,1/81)(40,1/81 - 1/486)
                };
            \addplot [
                color=red,
                ]
                coordinates {
                (3.5 - 1.5*0.015,2/3 - 0.015)(4,1/3)
                };
            \end{groupplot}
        \end{tikzpicture}

        \renewcommand{\thefigure}{\ref*{fig:chore_online_general_size_plot_phi_j_star_i}}%
        \captionsetup{list=no}%
        \caption[]{Plot of the potential term $\smash{\Phi^{j^*}_i}$ as a function of the $\alpha$-\PropOne-slack. Machines start on the blue graph and switch to the red graph when they receive a job larger than $\alpha \cdot p_{j^*} / k_{j^*}$.}
        \addtocounter{figure}{-1}
        \label{fig:chore_online_general_size_plot_phi_j_star_i_appendix}
    \end{figure}
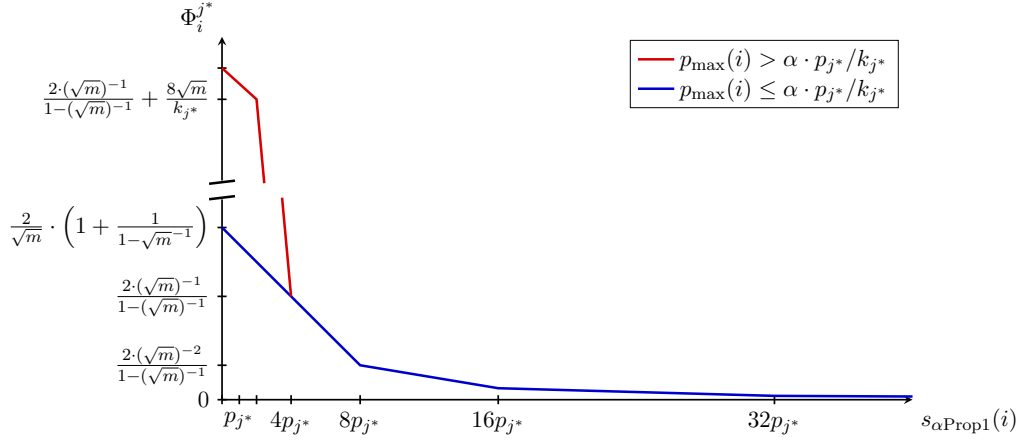

    Similar to the proof of \cref{thm:chore_online_unit_size_greedy_ratio}, we show that $\Phi^{j^*}$ does not increase when the previous jobs are allocated and that the induction statement holds for $j^*$ as long as $\Phi^{j^*}$ has at most its initial value.

    To show that the induction statement holds as long as $\Phi^{j^*}$ has at most its initial value, we will use the penalty terms of the potential function that punish machines on which $j^*$ cannot shadow and whose $\alpha$-\PropOne-slack is too small (or barely large enough) for $j^*$ to fit.

    In order to prove that the potential does not increase when previous jobs are allocated, we have to show that every decrease in $\alpha$-\PropOne-slack and every activation of a penalty term is compensated for by the increase in proportional shares on the respective other eligible machines, using the induction hypothesis for previous jobs that went into the shadow case.
    Intuitively, this is achieved as follows:

    The $\alpha$-\PropOne-slack of a machine $i$ decreases only in the fit case where it decreases by at most a factor $2$. Since each segment where the potential function has constant slope spans at least a factor $2$, the increase in $\smash{\Phi^{j^*}_i}$ per decrease in $s_{\alpha\PropOne}(i)$ can become larger only by at most a factor $\smash{2\sqrt{m} = \frac{\alpha}{4}}$ (even if the penalty term is active).
    By the behavior of the algorithm, the other eligible machines have at most the same $\alpha$-\PropOne-slack as machine $i$ before the job is allocated.
    Therefore, one quarter of their collective potential decrease per increase in $\alpha$-\PropOne-slack is already sufficient to compensate for the increase in $\smash{\Phi^{j^*}_i}$.

    If a potential term $\smash{\Phi^{j^*}_i}$ increases because a job $j$ activates the penalty term, then $j$ has to be sufficiently large while the $\alpha$-\PropOne-slack of machine $i$ is at most $4 p_{j^*}$.
    To be precise, the processing time of $j$ has to be larger than $\alpha \cdot p_{j^*} / k_{j^*}$, which is proportional to the size of the penalty term.

    By the behavior of our algorithm (both in the fit case and in the shadow case) and the induction assumption applied to $j$ (only in the shadow case), at least half of the eligible machines also have $\alpha$-\PropOne-slacks smaller than $4 p_{j^*}$.
    We will show that these machines provide enough potential decrease per increase in $\alpha$-\PropOne-slack to pay for the full penalty term of machine $i$ using the proportional share they receive from~$j$.

    We first bound the initial value of $\smash{\Phi^{j^*}}$.
    The value of $\smash{\Phi^{j^*}_i}$ for any machine $i$ before the first job arrives is
    \begin{align*}
        \Phi^{j^*}_i &= \int_{0}^\infty (2 \sqrt{m})^{- \rho(t)} dt\\
        &= \int_{0}^8 (2 \sqrt{m})^{- \rho(t)} dt + \sum_{l = 0}^\infty \int_{8 \cdot 2^l}^{8 \cdot 2^{l+1}} (2 \sqrt{m})^{- \rho(t)} dt\\
        &= \frac{4}{\sqrt{m}} + \sum_{l = 0}^\infty 8 \cdot 2^l \cdot (2 \sqrt{m})^{- (l + 2)}\\
        &= \frac{2}{\sqrt{m}} \cdot \left(2 + \sum_{l = 0}^\infty \sqrt{m}^{- (l + 1)}\right)\\
        &= \frac{2}{\sqrt{m}} \cdot \left(1 + \frac{1}{1 - \sqrt{m}^{-1}}\right).
    \end{align*}
    Summing over all $m$ machines, using $m \geq 4$ we can bound the total potential in the beginning as follows:
    \begin{align*}
        \Phi^{j^*} &= 2 \cdot \sqrt{m} \cdot \left(1 + \frac{1}{1 - \sqrt{m}^{-1}}\right)\\
        &\leq 2 \cdot \sqrt{m} \cdot \left(2 + 2 \cdot \sqrt{m}^{-1}\right)\\
        &= 4 \sqrt{m} + 4.
    \end{align*}
    Next, we lower bound the value of $\Phi^{j^*}$ in any configuration in which at least $\frac{k_{j^*}}{2}$ machines have a job that is larger than $\frac{\alpha \cdot p_{j^*}}{k_{j^*}}$ while their $\alpha$-\PropOne-slack is less than $2 \cdot p_{j^*}$.
    The potential term of every such machine $i$ includes the full penalty term and is therefore at least
    \begin{align*}
        \Phi^{j^*}_i &\geq \int_{2}^4 4 \cdot \frac{\sqrt{m}}{k_{j^*}} dt + \int_{4}^\infty (2 \sqrt{m})^{- \rho(t)} dt\\
        &\geq 8 \cdot \frac{\sqrt{m}}{k_{j^*}} + \int_{4}^8 (2 \sqrt{m})^{- \rho(t)} dt\\
        &= 8 \cdot \frac{\sqrt{m}}{k_{j^*}} + \frac{2}{\sqrt{m}}.
    \end{align*}
    Summing over the $\frac{k_{j^*}}{2}$ machines, we get
    \begin{align*}
        \Phi^{j^*} &\geq \frac{1}{2} \cdot k_{j^*} \cdot \left(8 \cdot \frac{\sqrt{m}}{k_{j^*}} + \frac{2}{\sqrt{m}}\right)\\
        &\geq 4 \sqrt{m} + \frac{k_{j^*}}{\sqrt{m}}\\
        &> 4 \sqrt{m} + 4
    \end{align*}
    where we used $k_{j^*} > \alpha \geq 4 \sqrt{m}$.
    This shows that the lemma statement holds in every configuration in which $\Phi^{j^*}$ has at most its initial value.

    It remains to show that $\Phi^{j^*}$ does not increase when a previous job is allocated.
    To this end, let $j$ be a previous job and let $i$ be the machine $j$ is allocated to.

    For any eligible machine $i' \in \M(j)$,
    we denote by
    \begin{itemize}
        \item $s_{\text{old}}(i')$ its $\alpha$-\PropOne-slack before $j$ arrives, by
        \item $\widehat{s}(i')$ its $\alpha$-\PropOne-slack including the proportional share from job $j$ but not the additional load and by
        \item $s_{\text{new}}(i')$ its $\alpha$-\PropOne-slack including the proportional share from job $j$ as well as the additional load if $j$ is allocated to $i'$.
    \end{itemize}
    In particular, we have
    \begin{align*}
        \widehat{s}(i') &= s_{\text{old}}(i') + \frac{\alpha \cdot p_j}{k_j}.
    \end{align*}
    Similarly, we denote by $\smash{(\Phi^{j^*}_{i'})_{\text{old}}}$, $\smash{\widehat{\Phi^{j^*}_{i'}}}$ and $\smash{(\Phi^{j^*}_{i'})_{\text{new}}}$ the potential term of machine $i'$ computed with the respective $\alpha$-\PropOne-slacks where $p_{\max}(i')$ is updated together with the load.
    Moreover, let $(\Phi^{j^*})_{\text{old}}$, $\widehat{\Phi^{j^*}}$ and $(\Phi^{j^*})_{\text{new}}$ be the sums of the respective potential terms.

    By the induction assumption, either $j$ fit on $i$ and we have $2 \cdot{} p_j \leq \widehat{s}(i)$ or $j$ shadowed on $i$ and we have $p_{\max}(i') \leq \frac{\alpha \cdot p_j}{k_j}$ for more than half of the machines eligible for $j$.
    We will treat these two cases separately.
    \paragraph{Case 1 (Fit): $2 \cdot{} p_j \leq \widehat{s}(i)$.}~\\
    Let
    \begin{align*}
        \widehat{\rho} &\coloneqq \rho\left(\frac{\widehat{s}(i)}{p_{j^*}}\right).
    \end{align*}
    Then since $\rho$ is nondecreasing and the algorithm assigns $j$ to an eligible machine with maximum $\widehat{s}(i)$, we have
    \begin{align*}
        \widehat{\rho} &\geq \rho\left(\frac{\widehat{s}(i')}{p_{j^*}}\right)
    \end{align*}
    for any eligible machine $i'$.

    Now the decrease in potential caused by the proportional share from $j$ on any machine~$i' \in \M(j)$ can be lower bounded as follows:
    \begin{align*}
        (\Phi^{j^*}_{i'})_{\text{old}} - \widehat{\Phi^{j^*}_{i'}} &\geq \int_{s_{\text{old}}(i') / p_{j^*}}^{\widehat{s}(i') / p_{j^*}} (2\sqrt{m})^{-\rho(t)} dt\\
        &\geq \int_{s_{\text{old}}(i') / p_{j^*}}^{\widehat{s}(i') / p_{j^*}} (2\sqrt{m})^{-\widehat{\rho}} dt\\
        &= \frac{\alpha \cdot p_j}{k_j \cdot p_{j^*}} \cdot (2\sqrt{m})^{-\widehat{\rho}}.
    \end{align*}
    Summed over all $k_j$ eligible machines, the potential decreases by at least
    \begin{align*}
        (\Phi^{j^*})_{\text{old}} - \widehat{\Phi^{j^*}} &\geq \frac{\alpha \cdot p_j}{p_{j^*}} \cdot (2\sqrt{m})^{-\widehat{\rho}}.
    \end{align*}
    On the other hand, by the inequality given by the fit case, we have $s_{\text{new}}(i) \geq \widehat{s}(i)/2$. Using this inequality and the fact that the slope of $\Phi^{j^*}_i$ on the interval $(0, 4)$ is bounded by~$-\frac{4 \sqrt{m}}{k_{j^*}}$, the increase in potential on machine $i$ is bounded by
    \begin{align*}
        (\Phi^{j^*}_{i})_{\text{new}} - \widehat{\Phi^{j^*}_{i}} &\leq \begin{cases}
            \displaystyle \int_{s_{\text{new}}(i) / p_{j^*}}^{\widehat{s}(i) / p_{j^*}} (2 \sqrt{m})^{- \rho(t)} dt + \frac{8 \sqrt{m}}{k_{j^*}} & p_j \geq \alpha \cdot \frac{p_{j^*}}{k_{j^*}} \wedge \widehat{s}(i) \leq 8 \cdot p_{j^*}\\[2ex]
            \displaystyle \int_{s_{\text{new}}(i) / p_{j^*}}^{\widehat{s}(i) / p_{j^*}} \frac{4 \sqrt{m}}{k_{j^*}} dt & p_j < \alpha \cdot \frac{p_{j^*}}{k_{j^*}} \wedge \widehat{s}(i) \leq 8 \cdot p_{j^*}\\[2ex]
            \displaystyle \int_{s_{\text{new}}(i) / p_{j^*}}^{\widehat{s}(i) / p_{j^*}} (2 \sqrt{m})^{- \rho(t)} dt & \widehat{s}(i) > 8 \cdot p_{j^*}
        \end{cases}\\[1ex]
        &\leq \begin{cases}
            \displaystyle \int_{s_{\text{new}}(i) / p_{j^*}}^{\widehat{s}(i) / p_{j^*}} (2 \sqrt{m})^{- 1} dt + \frac{p_j}{p_{j^*}} & p_j \geq \alpha \cdot \frac{p_{j^*}}{k_{j^*}} \wedge \widehat{s}(i) \leq 8 \cdot p_{j^*}\\[2ex]
            \displaystyle \frac{p_j}{p_{j^*}} \cdot \frac{4 \sqrt{m}}{k_{j^*}} & p_j < \alpha \cdot \frac{p_{j^*}}{k_{j^*}} \wedge \widehat{s}(i) \leq 8 \cdot p_{j^*}\\[2ex]
            \displaystyle \int_{s_{\text{new}}(i) / p_{j^*}}^{\widehat{s}(i) / p_{j^*}} (2 \sqrt{m})^{- \rho\left(\frac{s_{\text{new}}(i)}{p_{j^*}}\right)} dt & \widehat{s}(i) > 8 \cdot p_{j^*}
        \end{cases}\\[1ex]
        &\leq \begin{cases}
            \displaystyle \int_{s_{\text{new}}(i) / p_{j^*}}^{\widehat{s}(i) / p_{j^*}} (2 \sqrt{m})^{- 1} dt + \frac{p_j}{p_{j^*}} & \widehat{s}(i) \leq 8 \cdot p_{j^*}\\[2ex]
            \displaystyle \int_{s_{\text{new}}(i) / p_{j^*}}^{\widehat{s}(i) / p_{j^*}} (2 \sqrt{m})^{- \rho\left(\frac{\widehat{s}(i)}{2 \cdot p_{j^*}}\right)} dt & \widehat{s}(i) > 8 \cdot p_{j^*}
        \end{cases}\\[1ex]
        &\leq \begin{cases}
            \displaystyle \int_{s_{\text{new}}(i) / p_{j^*}}^{\widehat{s}(i) / p_{j^*}} (2 \sqrt{m})^{0} dt + \frac{p_j}{p_{j^*}} & \widehat{s}(i) \leq 8 \cdot p_{j^*}\\[2ex]
            \displaystyle \int_{s_{\text{new}}(i) / p_{j^*}}^{\widehat{s}(i) / p_{j^*}} (2 \sqrt{m})^{- (\widehat{\rho} - 1)} dt & \widehat{s}(i) > 8 \cdot p_{j^*}
        \end{cases}\\[1ex]
        &\leq 2 \cdot \frac{p_j}{p_{j^*}} \cdot (2 \sqrt{m})^{- (\widehat{\rho} - 1)}.
    \end{align*}
    To see the last inequality, observe that $\widehat{s}(i) \leq 8 \cdot p_{j^*}$ holds if and only if $\widehat{\rho} = 1$.

    In total, we get
    \begin{align*}
        (\Phi^{j^*})_{\text{new}} - (\Phi^{j^*})_{\text{old}} &\leq 2 \cdot \frac{p_j}{p_{j^*}} \cdot (2 \sqrt{m})^{- (\widehat{\rho} - 1)} - \frac{\alpha \cdot p_j}{p_{j^*}} \cdot (2\sqrt{m})^{-\widehat{\rho}}\\
        &= (4 \sqrt{m} - \alpha) \cdot \frac{p_j}{p_{j^*}} \cdot (2 \sqrt{m})^{- \widehat{\rho}}\\
        &\leq 0.
    \end{align*}
    This concludes case 1.
    \paragraph{Case 2 (Shadow): More than half of the eligible machines $i'$ (and in particular $i$) satisfy $p_{\max}(i') \leq \frac{\alpha \cdot p_j}{k_j}$.}~\\
    Recall that the $\alpha$-$\PropOne$-slack of machine $i$ does not decrease when job $j$ is allocated to machine~$i$.
    Therefore, $\Phi^{j^*}_i$ as well as $\Phi^{j^*}$ will not increase unless the penalty term was zero before allocating $j$ to $i$ and is positive afterwards.
    This bad case can only happen if $p_{\max}(i) < \frac{\alpha \cdot p_{j^*}}{k_{j^*}}$ before $j$ arrived, $p_j \geq \frac{\alpha \cdot p_{j^*}}{k_{j^*}}$ and $s_{\text{old}}(i) < 4 \cdot p_{j^*}$.
    We will therefore assume in the following that these inequalities hold.

    Since the algorithm chooses $i$ with maximum $\alpha$-\PropOne-slack among the eligible machines $i'$ for which $p_{\max}(i') \leq \frac{\alpha \cdot p_j}{k_j}$, the inequality $s_{\text{old}}(i) < 4 \cdot p_{j^*}$ implies that $s_{\text{old}}(i') < 4 \cdot p_{j^*}$ holds for all such machines.
    This implies that there are at least $\frac{k_j - 1}{2}$ eligible machines $i' \neq i$ with $s_{\text{old}}(i') < 4 \cdot p_{j^*}$.

    If $\frac{\alpha \cdot p_j}{k_j} \leq 4 \cdot p_{j^*}$, then $\widehat{s}(i') < 8 \cdot p_{j^*}$ for all such $i'$ and the proportional share from job $j$ decreases the potential of any such machine by at least
    \begin{align*}
        (\Phi^{j^*}_{i'})_{\text{old}} - \widehat{\Phi^{j^*}_{i'}} &\geq \int_{s_{\text{old}}(i') / p_{j^*}}^{\widehat{s}(i') / p_{j^*}} (2\sqrt{m})^{-\rho(t)} dt\\
        &= \int_{s_{\text{old}}(i') / p_{j^*}}^{\widehat{s}(i') / p_{j^*}} (2\sqrt{m})^{-1} dt\\
        &= \frac{\alpha \cdot p_j}{k_j \cdot p_{j^*}} \cdot (2 \sqrt{m})^{-1}\\
        &\geq \frac{4 \alpha}{k_j \cdot k_{j^*}}.
    \end{align*}
    Here, the last inequality uses $\alpha \geq 8 \sqrt{m}$ and $p_j \geq \frac{\alpha \cdot p_{j^*}}{k_{j^*}}$.

    Otherwise, we have $\frac{\alpha \cdot p_j}{k_j} > 4 \cdot p_{j^*}$. In this case, we obtain the above inequality as follows:
    \begin{align*}
        (\Phi^{j^*}_{i'})_{\text{old}} - \widehat{\Phi^{j^*}_{i'}} &\geq \int_{s_{\text{old}}(i') / p_{j^*}}^{\widehat{s}(i') / p_{j^*}} (2\sqrt{m})^{-\rho(t)} dt\\
        &\geq \int_{s_{\text{old}}(i') / p_{j^*}}^{s_{\text{old}}(i') / p_{j^*} + 4} (2\sqrt{m})^{-\rho(t)} dt\\
        &= \int_{s_{\text{old}}(i') / p_{j^*}}^{s_{\text{old}}(i') / p_{j^*} + 4} (2\sqrt{m})^{-1} dt\\
        &= \frac{2}{\sqrt{m}} \geq \frac{16}{\alpha} \geq \frac{16 \alpha}{k_j \cdot k_{j^*}} \geq \frac{4 \alpha}{k_j \cdot k_{j^*}}.
    \end{align*}
    In total, the potential decreases by at least
    \begin{align*}
        (\Phi^{j^*})_{\text{old}} - \widehat{\Phi^{j^*}} \geq \frac{k_j - 1}{2} \cdot \frac{4 \alpha}{k_j \cdot k_{j^*}} \geq \frac{\alpha}{k_{j^*}} \cdot \frac{k_j - 1}{k_j} \cdot 2 \geq \frac{\alpha}{k_{j^*}}.
    \end{align*}
    But the potential term for machine $i$ increases by at most the full penalty term which is bounded by
    \[
        (\Phi^{j^*}_{i})_{\text{new}} - \widehat{\Phi^{j^*}_{i}} \leq \frac{8 \sqrt{m}}{k_{j^*}} \leq \frac{\alpha}{k_{j^*}}.
    \]
    Thus, $j$ does not cause the potential to increase. This concludes case 2.

    We showed in all cases that the potential does not increase through any previous job $j$. As discussed above, this implies that the lemma statement holds for $j^*$ which completes the proof.
\end{proof}

\begin{proof}[Proof of \cref{thm:general_upper_bound}]
Choose $\alpha = 8\sqrt{m}$.
By \cref{lem:general-shadow-exists},
whenever the Fit case fails for a job with more than $\alpha$ eligible machines, there exists an eligible machine satisfying the
Shadow condition. Hence the algorithm is well-defined.

We first show that the $\alpha$-\PropOne-slacks computed from only jobs with more than $\alpha$ eligible machines stay nonnegative.
To this end, consider again the two cases in the algorithm.

In the Fit case, the updated slack of the chosen machine is at least $2p_j$
before allocating $j$. Allocating $j$ decreases its $\alpha$-\PropOne-slack by at
most $p_j$, so the slack remains nonnegative.

In the Shadow case, the chosen machine satisfies
\[
    p_{\max}(i) \leq \frac{\alpha p_j}{k_j}.
\]
After allocating $j$, the increase in $L_i-p_{\max}(i)$ is at most the previous
largest job $p_{\max}(i)$, while $\alpha\PS_i$ increases by
$\alpha p_j/k_j$. Thus the $\alpha$-\PropOne-slack does not decrease.
In particular, it does not become negative.

It remains to show that the $\alpha$-\PropOne{} guarantee does not break when reintroducing jobs with at most $\alpha$ eligible machines.
By the above results, all $\alpha$-\PropOne-slacks stay nonnegative when jobs with at most $\alpha$ eligible machines are not taken into account.
Adding a job $j$ with $k_j \leq \alpha$ back into the calculation cannot decrease the $\alpha$-\PropOne-slack of the machine $i$ to which $j$ is allocated as $\alpha\PS_i$ increases by $\alpha p_j/k_j$ which is at least as large as $p_j$.

Therefore, all machines have nonnegative $\alpha$-\PropOne-slack throughout the
algorithm and the final allocation is $\alpha$-\PropOne. Since
$\alpha=8\sqrt{m}$, this proves the theorem.
\end{proof}

\section{Omitted Proofs from \cref{sec:goods_setting} (Online Goods Setting)}
\label{app:goods-online}

\subsection{Greedy by deficit fails}
\label{app:deficit-greedy-fails}

\begin{lemma} \label{lem:goods_online_specific_lower_bound_slack_greedy}
    Consider the setting where goods arrive online and have to be allocated immediately and irrevocably.
    Allocating each good to an eligible agent $i$ for which the $1$-deficit $\PS_i - L_i$ is largest generally does not yield allocations satisfying $\frac{m}{8}$-\PropOne{}, even if all goods have utility $1$.
\end{lemma}
\begin{proof}
    Let $m$ (the number of agents) be even.

    Assume for now that there is a total ordering $\prec$ on the machines by which the greedy algorithm tiebreaks, i.e. whenever there are multiple eligible agents with largest deficit, it allocates the good to the largest agent according to $\prec$.

    Let the adversary do the following:
    As long as there exists a pair of agents $i_1, i_2$ with equal deficit, the adversary sends a good for which only $i_1$ and $i_2$ are eligible.

    Consider such a good $j$ with eligible agents $i_1$ and $i_2$. Wlog let $i_2 \prec i_1$. Then the good is allocated to $i_1$ which causes the deficit of $i_1$ to decrease by $\frac{1}{2}$ while the deficit of $i_2$ increases by $\frac{1}{2}$.

    The behavior of the deficits of the agents is therefore equivalent to that of a \emph{labeled chip-firing game} on $\frac{1}{2} \mathbb{Z}$ as described by \cite{hopkins2017sorting}. From the results of \cite{anderson1989disks} and \cite{hopkins2017sorting}, we get that this process terminates and that the deficits of the agents in the final configuration are
    \begin{align*}
        \frac{\frac{m}{2}}{2}, \frac{\frac{m}{2}-1}{2}, \dots, \frac{1}{2}, -\frac{1}{2}, -\frac{2}{2}, \dots, -\frac{\frac{m}{2}}{2}
    \end{align*}
    in the order of $\prec$.

    Let $i^*$ be the smallest agent according to $\prec$.
    Since the greedy algorithm tiebreaks away from $i^*$ and every good for which $i^*$ is eligible has another eligible agent with the same deficit as $i^*$, agent $i^*$ does not receive any good.
    Hence, their proportional share in the end is $\PS_{i^*} = \frac{m}{4}$ while their utility after adding one good is only $1$.

    Thus, this allocation does not satisfy better guarantees than $\frac{m}{4}$-\PropOne.

    To extend this lower bound to the case when the greedy algorithm is equipped with a tiebreak rule, we describe a two-phase approach:

    In the first phase, the adversary induces a configuration in which all agents have different deficits by sending goods as described in the proof of \cref{thm:good_setting_online_lb}.

    Once such a configuration is established, the second phase consists of running the procedure described above while ignoring the utilities and proportional shares from the goods sent in the first phase.
    This causes the greedy algorithm to \textquote{tiebreak} with respect to the order given by the deficits after the first phase.

    As the first phase adds at most one unit of utility to each agent, the utility of $i^*$ after adding one good is at most $2$. As this good from the first phase also adds a strictly positive amount of proportional share, adding the first phase can only improve the ratio between the utility and the proportional share of $i^*$ by a factor strictly smaller than $2$. Hence, the final allocation does not satisfy $\frac{m}{8}$-\PropOne.
\end{proof}

\subsection{Proof of \cref{lem:goods-online-show-potential-difference-factor-negative} (technical Lemma used in the proof of \cref{thm:goods-unit-upper})}
\label{app:goods-technical-lemma}

\TechnicalLemmaGoodsOnline*
\begin{proof}
    Since $k \geq 2$, it holds that
    \[
        2^{\frac{1}{k}} \leq 1 + \frac{2 \sqrt{2} - 2}{k}
    \]
    which we can use to bound
    \begin{align*}
        (k - 1) \cdot 2^{\frac{1}{k}} + \frac{2^{\frac{1}{k}}}{m} - k &\leq \left(k - 1 + \frac{1}{m}\right) \cdot \left(1 + \frac{2 \sqrt{2} - 2}{k}\right) - k\\
        &= 2 \sqrt{2} - 3 + \frac{1}{m} - \frac{2 \sqrt{2} - 2}{k} + \frac{2 \sqrt{2} - 2}{m \cdot k}.
    \end{align*}
    For $k = 2$, we can use $m \geq 3$ to obtain
    \[
        2 \sqrt{2} - 3 + \frac{1}{m} - \frac{2 \sqrt{2} - 2}{k} + \frac{2 \sqrt{2} - 2}{m \cdot k} = \sqrt{2} - 2 + \frac{\sqrt{2}}{m} \leq \frac{4}{3} \sqrt{2} - 2 < 0.
    \]
    For $k \geq 3$, we use $m \geq k$ for the first occurrence and $m \geq 3$ for the second occurrence and get
    \begin{align*}
        2 \sqrt{2} - 3 + \frac{1}{m} - \frac{2 \sqrt{2} - 2}{k} + \frac{2 \sqrt{2} - 2}{m \cdot k} &\leq 2 \sqrt{2} - 3 + \frac{1}{k} - \frac{2 \sqrt{2} - 2}{k} + \frac{2 \sqrt{2} - 2}{3 \cdot k}\\
        &= 2 \sqrt{2} - 3 + \frac{1 + 2 - \frac{2}{3} - 2 \sqrt{2} + \frac{2}{3} \sqrt{2}}{k}\\
        &\leq 2 \sqrt{2} - 3 + \frac{\frac{7}{3} - \frac{4}{3} \sqrt{2}}{3} = \frac{14}{9} \sqrt{2} - \frac{20}{9} < 0.
    \end{align*}
\end{proof}

\subsection{Proof of \cref{thm:log-m-upper-bound-online-goods} (upper bound $\log m$ for online goods)}
\label{app:ub-online-goods}

For convenience, we restate the main theorem of \Cref{sec:goods_setting}.

\ThmLogmGoods*
\begin{proof}
    In the case of one agent, the only allocation trivially satisfies \PropOne{}.

    For $m \geq 2$ agents, we analyze the performance of the bucketing algorithm that builds on the greedy deficit algorithm for equal-size goods and \Cref{thm:goods-unit-upper}.

    Upon the arrival of a good $j$ with utility $p_j$, we round its size up to the nearest power of $2$, denoted by $p'_j = 2^{\lceil \log_2 p_j \rceil}$. We maintain a separate instance for each subset of goods with equal rounded utilities.

    Because we round up by a factor of at most $2$, the proportional share $\PS_i$ of an agent in the original instance is related to the proportional share $\PS_{ik}$ in virtual instance $k$ such that $\sum_k \PS_{ik} = \Theta(\PS_i)$. Applying \Cref{thm:goods-unit-upper} to each virtual instance $k$ yields:
	\[
	L_{ik} + p_{ik}^{\max} \geq \Omega\left(\frac{1}{\log m}\right) \cdot \PS_{ik},
	\]
	where $p_{ik}^{\max}$ is the utility of the largest good in class $k$ for which agent $i$ is eligible and that is not allocated to $i$ (or $0$ if no such good exists). Summing over all classes $k$:
	\[
	L_i + \sum_k p_{ik}^{\max} \geq \Omega\left(\frac{1}{\log m}\right) \cdot \PS_i.
	\]
	Since the sum $\sum_k p_{ik}^{\max}$ is a geometric series dominated by the utility of the largest good in the original instance for which $i$ is eligible and that is not allocated to $i$, denoted $p_i^{\max}$, we have $\sum_k p_{ik}^{\max} = O(p_i^{\max})$. Incorporating this into the inequality, we obtain:
	\[
	L_i + O(p_i^{\max}) \geq \Omega\left(\frac{1}{\log m}\right) \cdot \PS_i.
	\]
	Absorbing the constant factors into the $\Omega(\cdot)$ notation yields the claimed $O(\log m)$–\PropOne{} guarantee.
\end{proof}

\section{Omitted proofs from \Cref{sec:envy} (Envy-based Fairness)}
\label{app:envy}

\subsection{Equivalence of Goods and Chores}
\label{app:goods=chores}

Identical valuations make the goods and chore settings equivalent for envy-based fairness. Relabeling each chore as a good of equal value preserves the allocation, and hence every eligibility constraint, so the strong and weak notions correspond exactly across the two settings.

\PropEquivalenceGoodsChores*
\begin{proof}
Fix two machines $i,i'$. For both the strong and the weak notion, the jobs relevant for comparing $i$ with $i'$ in the chore instance are exactly the jobs relevant for comparing $i'$ with $i$ in the corresponding goods instance. Replacing each chore by a good of the same magnitude reverses the direction of envy: machine $i$ envies machine $i'$ in the chore instance if and only if machine $i'$ envies machine $i$ in the goods instance.
The same correspondence holds after removing one relevant item, and after removing any relevant item. Hence EF, EF1, and EFX are preserved under this transformation, for both the strong and the weak notion.
\end{proof}

\subsection{Strong Envy}
\label{app:strong-ef1}

\ThmNoEFOneChore*
\begin{proof}
    Consider a first set $\N_1$ of $7$ jobs with processing time $1$ that have the following binary eligibility matrix, where each row represents a job, each column represents a machine and a machine is eligible for a job if and only if the corresponding entry in the matrix is $1$:
    \begin{align*}
        \begin{pmatrix}
            $1$ & $0$ & $1$ & $1$ & $1$ & $0$ & $0$ \\
            $0$ & $1$ & $0$ & $1$ & $1$ & $1$ & $0$ \\
            $0$ & $0$ & $1$ & $0$ & $1$ & $1$ & $1$ \\
            $1$ & $0$ & $0$ & $1$ & $0$ & $1$ & $1$ \\
            $1$ & $1$ & $0$ & $0$ & $1$ & $0$ & $1$ \\
            $1$ & $1$ & $1$ & $0$ & $0$ & $1$ & $0$ \\
            $0$ & $1$ & $1$ & $1$ & $0$ & $0$ & $1$
        \end{pmatrix}
    \end{align*}
    This incidence structure is the complement of the Fano plane, also known as the unique symmetric $2$-$(7,4,2)$ design and the biplane of order $2$ (see for example \cite{assmus1992designs}).
    We will use the following two of its properties:
	\begin{compactitem}
    \item every pair of jobs from $\N_1$ has exactly two common eligible machines; and
    \item every pair of machines is eligible for exactly two common jobs from $\N_1$.
	\end{compactitem}
	This implies that unless every machine is allocated exactly one job from $\N_1$, there is a pair of machines of which one machine has two jobs for which both machines are eligible while the other machine has no such jobs:
    Consider a machine $i$ taking two jobs $j, j' \in \N_1$ and consider the unique other machine $i'$ that is also eligible for both $j$ and $j'$. Then $j$ and $j'$ are the only two jobs from~$\N_1$ for which both $i$ and $i'$ are eligible, so $i'$ cannot have any such job.

    Now, add a second set $\N_2$ of two more jobs with processing time $p \in (\frac{1}{2}, 1)$ for which all machines are eligible. We claim that there is no allocation of all $9$ jobs that satisfies EF1.

    First, consider any allocation in which there is a machine $i$ that is assigned at least two jobs from $\N_1$. By the above argument, there is another machine $i'$ that is eligible for at least two of these jobs and that has no job from $\N_1$ for which machine $i$ is also eligible.
    Then $i$ will envy $i'$ even after removing one of its jobs unless $i'$ receives both jobs from $\N_2$. But since every machine is eligible for these two jobs, $i'$ then envies all machines that don't have a job for which $i'$ is eligible, even after removing one of the jobs from $i'$. Since there are $6$ other machines while $i'$ is eligible for only $4$ jobs from $\N_1$ (two of which are on machine $i$), such machines exist.

    It remains to check allocations in which each machine receives exactly one job from $\N_1$.

    If both jobs from $\N_2$ are assigned to the same machine $i$, then $i$ strongly envies every machine~$i'$ that does not envy~$i$ with respect to only the jobs in $\N_1$, even after removing the unit-size job assigned to $i$ from the first set. Since at most three other machines have a job from $\N_1$ for which $i$ is eligible, such a machine $i'$ exists. Thus strong EF1 is violated.

    Otherwise, the jobs from $\N_2$ are allocated to two different machines. To handle this case, we need to explore the structure of the allocation of the jobs in $\N_1$ a bit further. To this end, let $G_{\N_1}$ be the envy graph after adding only the jobs from $\N_1$.

    We first show that $G_{\N_1}$ satisfies
    \[\vert \delta^-_{G_{\N_1}}(i) \vert = \vert \delta^+_{G_{\N_1}}(i) \vert\]
    for all machines $i$:

    Let $i$ be some machine and let $j$ be the job allocated to machine $i$. Then $i$ envies exactly the machines $i'$ that are eligible for $j$ and whose job is incompatible with $i$. Since $j$ has $4$ eligible machines including $i$, the number of such machines is
    \begin{align*}
        3 - \vert \{\text{$i' \in \M(j): i' \neq i$, $i$ is eligible for the job assigned to $i'$}\} \vert.
    \end{align*}
    On the other hand, $i$ is envied by exactly the machines that are not eligible for $j$ and whose job is compatible with $i$. Since $i$ is eligible for $4$ jobs including $j$, the number of such machines is
    \begin{align*}
        &\phantom{{}={}} 3 - \vert \{\text{$j' \in \N_1: j' \neq j$, $i \in \M(j')$, the machine $i'$ to which $j'$ is assigned is eligible for $j$}\} \vert\\
        &= 3 - \vert \{\text{$i' \in \M(j): i' \neq i$, $i$ is eligible for the job $j'$ assigned to $i'$}\} \vert
    \end{align*}
    and thus equal to the number of machines that $i$ envies.

    Next, we want to show that $G_{\N_1}$ has at most one isolated vertex. To this end, assume for contradiction that $G_{\N_1}$ has two isolated vertices $i$ and $i'$. Since $i$ and $i'$ do not envy each other, $i'$ is eligible for the job $j$ allocated to $i$ if and only if $i$ is eligible for the job $j'$ allocated to $i'$.

    If $i$ and $i'$ are eligible for both $j$ and $j'$, let $i_2$ and $i_3$ be the other two machines eligible for $j$ and let $i_2'$ and $i_3'$ be the other two machines eligible for $j'$. Since $j$ and $j'$ share exactly two eligible machines which are $i$ and $i'$, all $6$ named machines are distinct. Let $\widehat{i}$ be the $7$th machine.

    Consider the job $j_2$ allocated to machine $i_2$. Since $i$ does not envy $i_2$, $i$ has to be eligible for $j_2$. As $j_2$ and $j$ share exactly two eligible machines which are $i$ and $i_2$, $j_2$ cannot be compatible with $i_3$ and $i'$. Moreover, at most one of $i_2'$ and $i_3'$ is eligible for $j_2$ as $j_2$ would otherwise share the three eligible machines $i$, $i_2'$ and $i_3'$ with $j'$. But then the three machines not eligible for $j_2$ have to be $i'$, $i_3$ and one of $i_2'$ and $i_3'$. Therefore, ~$\widehat{i}$ has to be eligible for $j_2$.

    By the same arguments, it follows that $\widehat{i}$ is also eligible for $j_3$, $j_2'$ and $j_3'$. But then all four jobs for which $\widehat{i}$ is eligible are allocated to other machines, contradicting the fact that each machine has one job.

    If $i'$ is not eligible for $j$ and $i$ is not eligible for $j'$, let $i^*_1$ and $i^*_2$ be the two eligible machines shared by $i$ and $i'$. Let $i_2$ be the fourth machine eligible for $j$ and let $i_2'$ be the fourth machine eligible for $j'$. Again, all $6$ named machines are distinct. Let $\widehat{i}$ be the $7$th machine.

    Consider the job $j^*_1$ allocated to $i^*_1$. Since both $i$ and $i'$ do not envy $i^*_1$, both of them have to be eligible for $j^*_1$. Then as $j$ and $j^*_1$ share the eligible machines $i$ and $i^*_1$, $i_2$ and $i^*_2$ cannot be eligible for $j^*_1$. Moreover, as $j'$ and $j^*_1$ share the eligible machines $i'$ and $i^*_1$, $i_2'$ cannot be eligible for $j^*_1$. Therefore, the three machines not eligible for $j^*_1$ are $i_2$, $i_2'$ and~$i^*_2$. It follows that $\widehat{i}$ is eligible for $j^*_1$.

    By the same arguments, $\widehat{i}$ is also eligible for $j^*_2$. But then $j^*_1$ and $j^*_2$ share the three eligible machines $i$, $i'$ and $\widehat{i}$, which is impossible given the structure of the instance.
    We conclude that $G_{\N_1}$ cannot have more than one isolated vertex.

    Using this fact together with our above result that $\vert \delta^-_{E(G_{\N_1})}(i) \vert = \vert \delta^+_{E(G_{\N_1})}(i) \vert$ for every machine $i$, we get that all machines but at most one lie on directed envy cycles in $G_{\N_1}$.
    Moreover, since no pair of machines can envy each other, every envy cycle is of length at least $3$.

    Now let $i_1$ and $i_2$ be the two machines receiving the jobs from $\N_2$.
    Since at least one of $i_1$ and $i_2$ is on an envy cycle of length at least $3$ in $G_{\N_1}$, at least one of $i_1$ and $i_2$ envies some other machine $i' \in \M \setminus \{i_1, i_2\}$ in $G_{\N_1}$.
    W.l.o.g.\ let this be $i_1$.
    Then $i_1$ is not eligible for the job from $\N_1$ allocated to $i'$ as $i_1$ would otherwise not envy $i'$ in $G_{\N_1}$.
    Since $i'$ also does not receive a job from $\N_2$, it does not have any job for which $i_1$ is eligible.
    Thus, $i_1$ envies $i'$ even after removing its larger job.

    In all cases, there is a pair of machines such that one strongly envies the other even after removing its largest job. Therefore, no allocation of these jobs satisfies EF1.
\end{proof}

The preceding theorem shows that strong EF1 may fail to exist offline when
processing times are general.
In the online setting with general processing times, even approximation does not restore attainability.
For $\alpha>0$, an allocation satisfies \emph{strong
$\alpha$-EF1} if, for every pair $i,i'\in\M$, either $i$ does not strongly envy $i'$, or there is a job $j\in C(i,i')$ assigned to $i$ such that
$L_{i,i'}-p_j\le \alpha L_{i',i}$. Already for
$m=3$ machines, no online algorithm can guarantee strong $\alpha$-EF1 for any
$\alpha>0$. %

\ThmNoAlphaEFOne*
\begin{proof}
    Fix $\alpha>0$, and let $p>\alpha$, say $p \coloneqq \alpha+1$. Let the machines be $1,2,3$. %

    Consider the following sequence of jobs: First, for each pair of machines, release one unit-size job eligible exactly for this pair. Our target assignment is as follows:
    \begin{compactitem}
        \item the job for which machines $1$ and $2$ are eligible should be assigned to machine $1$;
        \item the job for which machines $2$ and $3$ are eligible should be assigned to machine $2$;
        \item the job for which machines $1$ and $3$ are eligible should be assigned to machine $3$.
    \end{compactitem}

For every pair whose unit-size job was not assigned according to the target assignment, release one additional job of processing time $p$ eligible exactly for the same pair. Such a job must be assigned to the target machine of that pair: otherwise, the other machine has two jobs for which the target machine is also eligible, while the target machine has none compatible with the other machine. Even after removing one job from the other machine, its restricted load is at least $1$, whereas the target machine has restricted load $0$, contradicting strong $\alpha$-EF1.

Thus, after these correction jobs, for each pair of machines, the load of the target machine of this pair restricted to jobs for which both machines from the pair are eligible is larger than the load of the other machine restricted to these jobs by a factor of at least $p$.

Finally, release one job of processing time $p$ for which all three machines are eligible. Let $i$ be the machine receiving this final job. In the target assignment, machine $i$ is the target machine for exactly one pair; let $i'$ be the other machine in that pair.

Since all jobs released so far have size at most $p$, this last job is a largest job allocated to $i$ for which $i'$ is also eligible. After removing this job, we are in the same situation as before where the load of machine $i$ restricted to jobs for which both $i$ and $i'$ are eligible is larger than the load of $i'$ restricted to these jobs by a factor of at least $p > \alpha$. This violates strong $\alpha$-EF1.
\end{proof}

\subsection{Weak Envy}
\label{app:weak-envy}

Consider the equivalent of Algorithm 3 of Elkind et al.~\cite{ELL+25temporal}:
assign each arriving job to an eligible machine of minimum current load. Once a machine has larger load than another, any later job assigned to it by this rule must be incompatible with the other machine, and is therefore ignored by weak envy.

\ThmWeakEFOneGreedy*
\begin{proof}
    For the first statement, we use the above algorithm and arguments similar to \cite{ELL+25temporal}. %
    Consider two machines $i$ and $i'$ such that $i$ weakly envies $i'$. Consider the last time point at which $i$ did not envy $i'$ and let $j$ be the next job which is then allocated to machine $i$ and causes $i$ to envy $i'$ until the end of time. Observe that while $i$ envies $i'$, $i'$ always has smaller load than $i$. By the greedy criterion, $i'$ cannot be eligible for any job after $j$ which is allocated to $i$. Therefore, these jobs do not count for our envy notion.
    We conclude that machine $i$ does not envy machine $i'$ after removing job $j$ because $i$ did not envy $i'$ before being assigned $j$ and the later jobs do not make the envy worse.
    This shows that the allocation satisfies EF1.

    To compute an allocation satisfying weak EFX offline, run the above algorithm with the jobs sorted in order of non-increasing processing times. Then the job $j$ in the above argument is the smallest job on machine $i$ for which $i'$ is eligible which implies that the allocation computed in this way satisfies weak EFX.
\end{proof}

The positive online result does not extend to EFX.
\ThmNoWeakEFX*
\begin{proof}
    We use the instance and argument from \cite{ELL+25temporal}.
    For completeness, we give a full proof for our setting.

    Consider two machines and three jobs for which both machines are eligible. The first two jobs have processing time $1$, and the third job has processing time $2$.

	To maintain weak EFX after the first two jobs, an online algorithm must assign them to different machines: otherwise, one machine has load $2$ and the other has load $0$, and removing either unit-size job still leaves positive envy.

	The third job is then assigned to one of the two machines. This machine has load $3$, while the other machine has load $1$. Even after removing the job of processing time $1$ from the first machine, its remaining load is $2>1$. Thus, weak EFX is violated.
\end{proof}

\section{Impossibility Results for Scheduling Objectives and General Valuations}
\label{app:impossibility}

\subsection{General Valuations Rule Out Sublinear Online Guarantee}
\label{app:impossibility-general-valuations}

With general valuations, job $j$ has a machine-dependent processing time $p_{ij}$ for every $i \in M(j)$, and we define the proportional share analogously as $\PS_i := \sum_{j : i \in M(j)} p_{ij}/k_j$.

\begin{observation}
    No deterministic online algorithm can guarantee allocations satisfying $o(m)$-\PropOne{} in restricted-assignment chore allocation with general valuations.
\end{observation}
\begin{proof}
    Fix a deterministic online algorithm. Consider the following adversarial instance:

    The adversary first sends a total of $\frac{m}{2}$ jobs such that the eligible machines for the $\ell$-th job are exactly those machines that did not receive any of the first $\ell-1$ jobs. On these machines, the $\ell$-th job has processing time $2^\ell$.

    At this stage, since each of these jobs has more than $\frac{m}{2}$ eligible machines, the proportional share of the machine receiving the $\ell$-th job is smaller than $\frac{4 \cdot 2^{\ell}}{m}$ while the load of this machine is $2^\ell$.

    The adversary now releases one final job whose eligible machines are exactly those machines that received one of the first $\frac{m}{2}$ jobs. On the machine $i$ that received the $\ell$-th job, this job has processing time $p_{ij} = 2^\ell$.

    Consider the machine $i^*$ receiving this job. Let $\ell^*$ be the index of the first job allocated to $i^*$. Then $i^*$ has two jobs of processing time $2^{\ell^*}$. Hence, the load of $i^*$ without its largest job is $2^{\ell^*}$.

    On the other hand, the proportional share of $i^*$ is smaller than $\frac{6 \cdot 2^{\ell^*}}{m}$. Thus, $i^*$ violates $\alpha$-\PropOne{} for any $\alpha \leq \frac{m}{6}$.
\end{proof}

\subsection{Fairness Can Be Incompatible with Makespan}
\label{app:impossibility-simultan-scheduling}

\begin{observation}
    In restricted-assignment chore allocation, adding the requirement to satisfy $\alpha$-\PropOne{} for some $\alpha \geq 1$ may increase the optimum makespan by a factor $\Omega(\frac{m}{\alpha})$, even for unit-size jobs.
\end{observation}
\begin{proof}
    Let $d := \lfloor 2\alpha \rfloor \geq 2$, so that $d \leq 2\alpha < d + 1$.
    If $m < 3d + 1$, then $m \leq 3d \leq 6\alpha$, so $\frac{m}{\alpha} = O(1)$ and the claim is trivial.
    Otherwise, let $\beta \geq 1$ be the largest integer with $d \cdot (2\beta + 1) + 1 \leq m$.

    Let there be $d$ designated machines, $N := 2\beta d + 1$ non-designated machines, and $m - d - N \geq 0$ idle machines that are not eligible for any job.
    For every non-designated machine $i$, let there be $2$ jobs with processing time $1$ whose eligible machines are all designated machines and~$i$.
    By Theorem~\ref{thm:offline-propone}, an allocation satisfying \PropOne{}, and hence $\alpha$-\PropOne{}, exists.

    Every job has $d + 1$ eligible machines, so the proportional share of every non-designated machine~$i$ is $\PS_i = \frac{2}{d + 1}$.
    If $i$ received both of its jobs, then
    \[
        L_i - p_{\max}(i) = 1 > \frac{2\alpha}{d + 1} = \alpha \cdot \PS_i,
    \]
    since $d + 1 > 2\alpha$.
    Therefore, every allocation satisfying $\alpha$-\PropOne{} allocates at most one job to each non-designated machine.
    Hence at least $N = 2\beta d + 1$ jobs are allocated to the $d$ designated machines.
    This leads to a makespan of at least
    \[
        \left\lceil \frac{2\beta d + 1}{d} \right\rceil = 2\beta + 1 .
    \]
    In contrast, allocating both jobs of every non-designated machine to that machine achieves a makespan of $2$.
    Thus, requiring $\alpha$-\PropOne{} increases the optimum makespan by a factor of at least $\frac{2\beta + 1}{2}$.

    It remains to relate $\beta$ to $m$ and $\alpha$.
    By the maximality of $\beta$, we have $d \cdot (2\beta + 3) + 1 > m$, that is, $2\beta + 1 > \frac{m - 1}{d} - 2$.
    Since $m \geq 3d + 1$ implies $\frac{m-1}{d} \geq 3$, we obtain $2\beta + 1 > \frac{m-1}{3d} \geq \frac{m-1}{6\alpha}$.
    Hence the factor is at least $\frac{m-1}{12\alpha} = \Omega(\frac{m}{\alpha})$.
\end{proof}

\printbibliography

\end{document}